\documentclass[11pt,a4paper]{article}
\usepackage{fullpage}
\usepackage{graphicx}
\usepackage{amsmath,amsfonts,amssymb,amsthm}
\usepackage{hyperref}
\usepackage{cleveref}
\usepackage{enumitem}

\theoremstyle{plain}
\newtheorem{theorem}{Theorem}[section]
\newtheorem{lemma}[theorem]{Lemma}
\newtheorem{claim}[theorem]{Claim}
\newtheorem{corollary}[theorem]{Corollary}
\newtheorem{observation}[theorem]{Observation}

\theoremstyle{definition}
\newtheorem{definition}{Definition}[section]

\usepackage{xpatch}
\xpatchcmd{\proof}{\itshape}{\proofheadfont}{}{}
\newcommand{\proofheadfont}{\bfseries}
\makeatletter
\xpatchcmd{\proof}{\@addpunct{.}}{\@addpunct{:}}{}{}
\makeatother

\newcommand{\congest}{${\mathsf{CONGEST}}$}
\newcommand{\LCA}{\operatorname{LCA}}
\newcommand{\poly}{{\rm poly}}
\newcommand{\cA}{{\cal A}}
\newcommand{\Sketch}{\mathsf{sketch}}
\newcommand{\CanSketch}{\mathsf{csketch}}
\newcommand{\LCALabel}{\mathsf{anc}}
\newcommand{\LDS}{\mathsf{LDS}}
\newcommand{\LA}{\mathsf{LA}}
\newcommand{\LD}{\mathsf{LD}}
\newcommand{\ID}{\mathsf{id}}
\newcommand{\EID}{\mathsf{eid}}
\newcommand{\UID}{\mathsf{uid}}
\newcommand{\dilation}{\mbox{\tt d}}
\newcommand{\congestion}{\mbox{\tt c}}
\newcommand{\depth}{\operatorname{depth}}
\newcommand{\Gates}{R}
\newcommand{\lcagates}{r}

\title{Near-Optimal Distributed Computation of Small Vertex Cuts}
\author{
Merav Parter\thanks{This project is funded by the European Research Council (ERC) under the European Union’s Horizon 2020 research and innovation programme (grant agreement No. 949083), and by the Israeli Science Foundation (ISF), grant No. 2084/18.}\\ 
\small Weizmann Institute \\
\small merav.parter@weizmann.ac.il
\and				
Asaf Petruschka \\
\small Weizmann Institute \\
\small asaf.petruschka@weizmann.ac.il 
}

\date{}

\begin{document}
\maketitle

\begin{abstract}
We present near-optimal algorithms for detecting small vertex cuts in the \congest\ model of distributed computing. 
Despite extensive research in this area, our understanding of the \emph{vertex} connectivity of a graph is still incomplete, especially in the distributed setting. To this date, all distributed algorithms for detecting cut vertices suffer from an inherent dependency in the maximum degree of the graph, $\Delta$. Hence, in particular, there is no truly sub-linear time algorithm for this problem, not even for detecting a \emph{single} cut vertex. 
We take a new algorithmic approach for vertex connectivity which allows us to bypass the existing $\Delta$ barrier.

As a warm-up to our approach, we show a simple $\widetilde{O}(D)$-round\footnote{Throughout the paper, we use the notation $\widetilde{O}$ to hide poly-logarithmic in $n$ terms.}
randomized algorithm for computing all cut vertices in a $D$-diameter $n$-vertex graph. This improves upon the $O(D+\Delta/\log n)$-round algorithm of [Pritchard and Thurimella, ICALP 2008].

Our key technical contribution is an $\widetilde{O}(D)$-round randomized algorithm for computing all cut \emph{pairs} in the graph, improving upon the state-of-the-art $O(\Delta \cdot D)^4$-round algorithm by [Parter, DISC '19]. Note that even for the considerably simpler setting of \emph{edge} cuts, currently $\widetilde{O}(D)$-round algorithms are known \emph{only} for detecting pairs of cut edges.

Our approach is based on employing the well-known linear graph sketching technique [Ahn, Guha and McGregor, SODA 2012] along with the heavy-light tree decomposition of [Sleator and Tarjan, STOC 1981]. Combining this with a careful characterization of the survivable subgraphs, allows us to determine the connectivity of $G \setminus \{x,y\}$ for every pair $x,y \in V$, using $\widetilde{O}(D)$-rounds. We believe that the tools provided in this paper are useful for omitting the $\Delta$-dependency even for larger cut values.
\end{abstract}
\newpage

\tableofcontents
\newpage

\section{Introduction and Our Contribution}
The vertex connectivity of the graph is a central concept in graph theory and extensive attention has been paid to developing algorithms to compute it in various computational models. Recent years have witnessed an enormous progress in our understanding of vertex cuts, from a pure graph theoretic perspective \cite{PettieY21} to many algorithmic applications  \cite{NanongkaiSY19,LiNPSY21,PettieY21,HeLW21}. Despite this exciting movement, our algorithmic toolkit for handling vertex cuts is still somewhat limited.   A large volume of the work, in the centralized setting, has focused on fast algorithms for detecting minimum vertex cuts of size at most $k$, for some small number $k$. Until recently, near-linear time algorithms where known only for $k\leq 2$ \cite{Tarjan72,HopcroftT73}. 
A sequence of recent breakthrough results \cite{ChenKLPPS,LiNPSY21, SaranurakY22} provide almost-linear time sequential algorithms for computing the vertex connectivity (even for large connectivity values).

As we discuss soon, the situation is considerably worse in distributed settings, where the problem is still fairly open already for $k=1$.
Throughout, we consider the \congest\ model \cite{Peleg:2000}. In this model, each node holds a 
processor with a unique and arbitrary ID of $O(\log n)$ bits, and initially only knows the IDs of its neighbors in the graph. The execution proceeds in synchronous rounds, where in each round, each node can send a message of size $O(\log n)$ to 
each of its neighbors. The primary complexity measure is the number of communication rounds. 
For $n$-vertex $D$-diameter graphs, Pritchard and Thurimella \cite{pritchard2011fast} presented a randomized algorithm for detecting a (single) cut vertex (a.k.a articulation point) within $O(D+\Delta/\log n)$ \congest\ rounds, where $\Delta$ is the maximum degree of the graph. The authors of \cite{pritchard2011fast} conclude their paper by noting:

\begin{quote}\cite{pritchard2011fast}
\emph{It would be interesting to know if our distributed cut vertex algorithm could be synthesized with the
cut vertex algorithm of \cite{thurimella1997sub} to yield further improvement. Alternatively, a lower bound showing that no
$O(D)$-time algorithm is possible for finding cut vertices would be very interesting.}
\end{quote}

No progress on the complexity of this problem has been done since then. For small cut values $k$, Parter \cite{parter2019small} employed the well-known fault-tolerant sampling technique \cite{weimann2013replacement,SP21} for detecting $k$ vertex cuts in $(\Delta \cdot D)^{\Theta(k)}$ deterministic rounds. Turning to approximation algorithms, for $k =\Omega(\log n)$, Censor-Hillel, Ghaffari and  Kuhn \cite{censor2014distributed} provided a $O(\log n)$ approximation for computing the \emph{value} of the vertex connectivity of the graph within $\widetilde{O}(D+\sqrt{n})$ rounds
% \footnote{Throughout the paper, we use the notation $\widetilde{O}$ to hide poly-logarithmic in $n$ terms.}.
The authors of \cite{censor2014distributed} also presented a lower bound of $\widetilde{\Omega}(D+\sqrt{n/k})$ $V$-\congest\ rounds. In the $V$-\congest\ model, each \emph{node} (rather than an edge) is restricted to send only $O(\log n)$ bits, in total, in every round. 
As shown in this paper, this lower bound does not hold in the standard \congest\ model.  

We follow the terminology from \cite{pritchard2011fast}: a \emph{cut vertex} is a vertex $x$ such that $G \setminus \{x\}$ is not connected. A \emph{cut pair} is a pair of vertices $x,y$ such that $G \setminus \{x,y\}$ is not connected. For brevity,
% we call these objects \emph{small cuts}.
we refer to cut vertices and cut pairs as \emph{small cuts}.
Our main results in this paper are near-optimal algorithms for detecting these small cuts, in the sense that for every small cut, there is at least one vertex in the graph that learns it.
Our first contribution is in presenting a (perhaps surprisingly) simple randomized algorithm\footnote{As usual, all presented randomized algorithms in this paper have success guarantee of $1-1/n^c$, for any given constant $c>1$.} that can detect all cut vertices in the graph in $\widetilde{O}(D)$ rounds.  The edge-congestion\footnote{The edge congestion of a given algorithm is the worst-case bound on the total number of messages exchanged through a given edge $e$ in the graph.} of the algorithm is $\widetilde{O}(1)$ bits\footnote{We exploit this bounded congestion for detecting cut pairs.}.

\begin{theorem}\label{thm:distributed-single-cut}
There is a randomized algorithm that w.h.p.\ identifies all single cut-vertices in $G$ within $\widetilde{O}(D)$ rounds. The edge congestion is $\widetilde{O}(1)$. In the output, each vertex $x \in V$ learns if it is a cut vertex.
\end{theorem}

This settles the question raised in \cite{pritchard2011fast}. Our algorithm is based on the well-known \emph{graph-sketching} technique of Ahn, Guha and McGregor \cite{ahn2012analyzing}.  This technique has admitted numerous applications in the context of connectivity computation under various computational settings, e.g., \cite{kapron2013dynamic,kapralov2014spanners,GibbKKT15,KingKT15,MashreghiK18,GhaffariP16,DuanConnectivitySODA17,DoryP21}. Yet, to the best of our knowledge, it has not been employed before in the context of \congest\ algorithms for minimum vertex-cut computation. 

We then turn to consider the problem of detecting cut pairs. 
It has been noted widely in the literature that there is a sharp qualitative difference between a single failure and two failures. This one-to-two jump has been accomplished by now for a wide variety of fault-tolerant settings, e.g., reachability oracles \cite{Choudhary16}, distance oracles \cite{duan2009dual}, distance preservers \cite{parter2015dual,GuptaK17,ParterDISC20} and vertex-cuts \cite{HopcroftT73,BattistaT89,BattistaT96,GeorgiadisILP15}. While it is relatively easy to extend our algorithm of Theorem \ref{thm:distributed-single-cut} to detect cut pairs in $\widetilde{O}(D^2)$ rounds, providing a near-optimal complexity of $\widetilde{O}(D)$ rounds turns out to be quite involved. Our key technical contribution is:

\begin{theorem}\label{thm:distributed-cut-pairs}
There is a randomized algorithm that w.h.p.\ identifies all cut pairs in $G$ within $\widetilde{O}(D)$ rounds. For each cut pair $x,y$, either $x$ or $y$ learns that fact. 
\end{theorem}

We observe that even for the simpler problem of edge-connectivity (see Remark on the Edge Connectivity below), an $\widetilde{O}(D)$-round algorithm is currently only known for edge cuts of size at most two \cite{pritchard2011fast,Daga20}. Hence, we are now able to match the complexity of these two problems for small cut values.  Our algorithm is based on distinguishing between two structural cases depending on the locations of the cut pair $x,y$ in a BFS tree of $G$. The first case which we call \emph{dependent} handles the setting where the $x$ and $y$ have ancestry/descendant relations. The second \emph{independent} case assumes that $x$ and $y$ are not dependent, i.e., $\LCA(x,y)\notin \{x,y\}$, where $\LCA(x,y)$ is the lowest (or least) common ancestor of $x$ and $y$ in the BFS tree. Each of these cases calls for a different approach. For a more in-depth technical overview, see \Cref{sec:approach}. 

We believe that the tools provided in this paper could hopefully pave the way towards detecting larger vertex cuts with no dependency in the maximum degree $\Delta$ (as it is the case for the state-of-the-art algorithm in \cite{parter2019small}).

\emph{Remark on the Edge Connectivity.} It is widely known that in undirected graphs, vertex connectivity and vertex cuts are
significantly more complex than edge connectivity and edge cuts, for which now the following result are
known: an $\widetilde{O}(m)$-time centralized exact algorithm \cite{karger1999random,Ghaffari0T20,GawrychowskiMW20} and an $\widetilde{O}(D+\sqrt{n})$ exact distributed algorithms \cite{DoryEMN21}. For constant values of edge connectivity a $\poly(D)$-round algorithm is given in \cite{parter2019small}.

\subsection{Our Approach, in a Nutshell}\label{sec:approach}

We provide the key ideas of our algorithms. Our end goal is to simulate a connectivity algorithm in the graph $G \setminus \{x,y\}$ for \emph{every} pair of vertices $x,y \in V$. Note that this is not trivial already for a single $x,y$ pair as the diameter of the subgraph\footnote{Here we mean the \emph{strong} diameter of $G\setminus \{x,y\}$, i.e. the diameter of the graph induced by $G$ on $V \setminus \{x,y\}$.
Its \emph{weak} diameter, defined as the maximal distance \emph{in $G$} between any $u,v \in V \setminus \{x,y\}$, remains at most $D$, which we crucially exploit.} $G \setminus \{x,y\}$ might be as large as $\Omega(\Delta D)$, hence using on-shelf connectivity algorithms leads to a round complexity of $O(\min\{D+\sqrt{n}, \Delta D\})$. 
We bypass this $\Delta$ dependency by using the edges incident to the vertices $x,y$ as \emph{shortcuts}. Then, to minimize the congestion imposed by running possibly $n^2$ connectivity algorithms in parallel, we employ a preprocessing phase in which we collect \emph{graph-sketch} information (explained next) at each vertex $x$. This information allows each vertex $x$ to pinpoint a bounded number of cut-mate suspects. In addition, it allows $x$, in certain cases, to locally simulate connectivity queries without using further communication.

Throughout, let $T$ be a BFS tree rooted at some source $s$, and denote the $T$-paths by $\pi(\cdot,\cdot)$. 
We start by employing the well-known \emph{heavy-light tree decomposition} technique by Sleator and Tarjan \cite{SleatorT83}. This classifies the edges of $T$ into light and heavy edges. The useful properties are that each vertex $v$ has $O(\log n)$ light edges on its tree path $\pi(s,v)$, and is the parent of \emph{one} heavy edge, connecting $v$ to its unique heavy child. It is easy to compute this decomposition on $T$ in $\widetilde{O}(D)$ rounds. For a vertex $x$, let $T_x$ be the subtree of $T$ rooted at $x$. 

\paragraph{Basic Tools: Graph Sketches and the Bor\r{u}vka Algorithm.}
A \emph{graph sketch} of a vertex $v$ is a randomized string of $\widetilde{O}(1)$ bits that compresses $v$'s edges \cite{ahn2012analyzing}. The linearity of these sketches allows one to infer, given the sketches of subset of vertices $S$, an outgoing cut edge from $S$ to $V \setminus S$ (if exists) with constant probability.  A common approach for deducing the graph connectivity merely from the sketches of the vertices is based on the well-known Bor\r{u}vka algorithm \cite{Boruvka}. This algorithm works in $O(\log n)$ phases, where in each phase, from each growable component an outgoing edge is selected. All these outgoing edges are added to the forest, while ignoring cycles. Each such phase reduces the number of growable components by constant factor, thus within $O(\log n)$ phases a maximal spanning forest\footnote{A maximal spanning forest is defined as the union of spanning trees for all connected components.}
is computed.
Since this algorithm only requires the computation of outgoing edges it can be simulated using $O(\log n)$ independent sketches for each of the vertices. In our algorithms, we aggregate graph sketches over the BFS tree $T$ which allows each vertex $x$ to locally simulate Bor\r{u}vka in the graph $G \setminus \{x\}$. This is illustrated in our algorithm for detecting a single cut vertex, described next. 

\paragraph{Warm Up: Detecting Single Cut Vertices.}
Our algorithm starts by letting each vertex $v$ locally compute its individual $\Sketch_G(v)$. Then, by aggregating the sketches (using their linearity) from the leaf vertices to the root $s$ over the BFS tree $T$, each vertex $v$ learns its subtree-sketch $\Sketch_G (V(T_v))$. Once this is completed, it is easy to let each vertex $x \in V$ learn the $G$-sketch information of all the connected components in $T \setminus \{x\}$. We then show that $x$ can locally modify these $G$-sketches into $(G \setminus \{x\})$-sketches. At this point, the vertex $x$ can locally apply the Bor\r{u}vka algorithm in $G \setminus \{x\}$ and deduce if $G \setminus \{x\}$ is connected.

We now turn to consider the considerably more challenging task of detecting cut pairs. We classify these pairs into dependent and independent. 

\paragraph{Detecting Dependent Cut Pairs.} 
Our approach for the dependent case is based on designing algorithms $\{\cA_y\}_{y \in V}$, where $\cA_y$ detects all $xy$ cut pairs of the form\footnote{This is a critical point where only $x$ learns if $y \in \pi(s,x)$ is its cut-mate (by running Alg. $\cA_y$), but $y$ might not learn its descendant cut-mates, such as $x$.} $x \in V(T_y)$. We show that each such an algorithm $\cA_y$ can be designed in a way that sends a total of $\widetilde{O}(1)$ messages only along edges incident to $V(T_y)$, and runs in $\widetilde{O}(D)$ rounds. The standard random delay technique allows us then to schedule the execution of all $n$ algorithms $\{\cA_y\}_{y \in V}$ within $\widetilde{O}(D)$ rounds. At a high level, each algorithm $\cA_y$ is based on employing the single cut vertex detection algorithm in the graph $G \setminus \{y\}$. Our challenge is then twofold: first, the diameter of the graph $G \setminus \{y\}$ might be as large as $\Omega(\Delta D)$, and second, communication is restricted to use only edges incident to $V(T_y)$. We overcome these challenges by using $y$ as a coordinator, providing global computation services and communication shortcuts that essentially enables efficient simulation (in both dilation and congestion) of the algorithm for detecting single cut vertices in $G \setminus \{y\}$.

\paragraph{Detecting Independent Cut Pairs.}
The most technically involved case is where $x,y$ are independent, namely, \emph{do not} have ancestry relations in $T$. A-priori, the number of such potential cut-mates $y$ for a given vertex $x$ might be even linear in $n$. To filter out irrelevant options, the algorithm starts by computing at each vertex $x$ a tree $\widehat{T}_x$ that encodes the connectivity between $s$ and the vertices in $V_x=V(T_x) \setminus \{x\}$ in the graph $G \setminus \{x\}$. Let $\mathcal{C}_x$ denote the collection of maximal connected components in the graph $G[V_x]$. The tree $\widehat{T}_x$ is the union of paths of the form $\pi(s,u_C)\circ (u_C,v_C)$ for every component $C \in \mathcal{C}_x$, where $v_C$ is some representative vertex in $C$, and $u_C$ is outside the subtree $T_x$.
It is then easy to observe that in case $\hat{T}_x$ is non-empty, all independent cut-mates $y$ must appear on some path $\pi(s,u_C)$ for $C \in \mathcal{C}_x$.
For a given suspect $y$, we call the $\mathcal{C}_x$-components $C$ for which $y \in \pi(s,u_C)$, $y$-\emph{sensitive}. Our argument has the following structure.

\smallskip
\noindent\textbf{Multiple $xy$-Connectivity Algorithms, Under a Promise.}
For a fixed $xy$ pair, we design an algorithm $\mathcal{A}^P_{x,y}$ that determines the connectivity in $G \setminus \{x,y\}$ \emph{given} an $x$-$y$ path $\Pi_{x,y}$ (on which $x,y$ can exchange messages). The algorithm $\mathcal{A}^P_{x,y}$ has the special property that it sends messages either along $\Pi_{x,y}$, or along edges incident to a restricted subset of vertices in $T_x, T_y$, defined as follows.
Let $\LDS(x,y) \subset V(T_x)$ be the set of all vertices which are descendants of the \emph{light} children of $x$, and belong to a $y$-sensitive component in $\mathcal{C}_x$. The set $\LDS(y,x)$ is defined in an analogous manner.
$\mathcal{A}^P_{x,y}$ is then guaranteed to send $\widetilde{O}(1)$ messages only along the edges of $\Pi_{x,y}$ and the edges incident to $\LDS(x,y) \cup \LDS(y,x)$. This restriction is crucial in order to run multiple $\mathcal{A}^P_{x,y}$ algorithms, for distinct $xy$ pairs, in parallel. Using the properties of the heavy-light decomposition and our sensitivity definition, one can show that each vertex $w \in V$ belongs to the $\LDS(x,y)$ sets of at most $\widetilde{O}(D)$ pairs $xy$. The main challenge is in bounding the overlap between the $\Pi_{x,y}$ paths across different $xy$ pairs. 
We show that given a subset $Q \in V \times V$, the collection of $\{\mathcal{A}^P_{x,y} \mid (x,y) \in Q\}$ algorithms
can be scheduled in parallel in $\widetilde{O}(D)$ rounds using the random delay approach \cite{leighton1994packet,Ghaffari15}, given that following promise holds for $Q$:

\begin{quote}[\textbf{Promise}:]
\emph{There is a path collection $\mathcal{P}_Q=\{\Pi_{x,y} \mid (x,y)\in Q\}$ such that each path has length $O(D)$, and each edge appears on $\widetilde{O}(D)$ paths in $\mathcal{P}_Q$. }
\end{quote}

On a high level, each algorithm $\mathcal{A}^P_{x,y}$ works by letting $x$ and $y$ jointly simulate the Bor\r{u}vka algorithm in $G \setminus \{x,y\}$. The main challenge is that the communication is restricted to the edges incident to $\LDS(x,y) \cup \LDS(y,x)$, despite the fact that one should also take into account the remaining vertices in $T_x,T_y$, e.g., descendants of the \emph{heavy} children of $x,y$. In each Bor\r{u}vka phase, we maintain the invariant that $x,y$ jointly hold the sketches of connected subsets (called \emph{parts}) in $G \setminus \{x,y\}$, where we split the responsibility between $x,y$ in a careful manner. We mainly distinguish between parts that contain a heavy child of $x,y$ and the remaining \emph{light} parts that are contained in $\LDS(x,y) \cup \LDS(y,x)$. The merges of the light parts are implemented by using communication between vertices in $\LDS(x,y) \cup \LDS(y,x)$. The merges concerning the heavy parts are implemented by using the direct communication over the $\Pi_{x,y}$ path. Each such Bor\r{u}vka phase is implemented in $\widetilde{O}(D)$ rounds. At the end of the simulation, $x,y$ both learn whether $G \setminus \{x,y\}$ is connected. 

\smallskip
\noindent\textbf{Omitting the Promise by Classifying to Light and Heavy Pairs.}
While the promise clearly holds for $\widetilde{O}(D)$ pairs, it does not hold for all $\Omega(n^2)$ pairs, in general. Our approach is based on classifying the independent $xy$ pairs into two classes: \emph{light} and \emph{heavy}.
This classification is based on the trees $\widehat{T}_x, \widehat{T}_y$, as well as on the heavy-light decomposition of $T$. Informally, for a light pair $xy$, one can define a $\Pi_{x,y}$ path that intersects a light subtree of either $x$ or $y$.  
These paths can be shown to have a bounded overlap, hence satisfying the promise. Handling the heavy pairs is more involved. Here we take a mixed approach. We define a special subset of the heavy pairs for which the promise can be satisfied (called \emph{mutual pairs}). This subset is chosen carefully to guarantee the following, perhaps surprising, property: the remaining (non-mutual) heavy pairs $x,y$ can be decided locally, at either $x$ or $y$. 
Our key observation is that for a $xy$ heavy pair, the graph $G \setminus \{x,y\}$ is connected iff one of the heavy children of $x,y$ is connected to $s$ in $G \setminus \{x,y\}$. Hence, it is mainly essential for $x,y$ to collect a sketch information on the components of these heavy children in $\mathcal{C}_x,\mathcal{C}_y$.

\subsection{Preliminaries}\label{sec:prelim}

Throughout the paper, fix a connected $n$-vertex graph $G = (V, E)$, and a BFS tree $T$ for $G$ rooted at an arbitrary source $s \in V$. We denote the unique tree path from $u$ to $v$ by $\pi(u,v,T)$. When the tree $T$ is clear from context, we may omit it and write $\pi(u,v)$. We use the $\circ$ operator for path concatenation. An (undirected) edge between vertices $u$,$v$ is denoted by $(u,v)$.

\paragraph{Heavy-Light Tree Decomposition.}
We now present our heavy-light terminology, the notion of \emph{compressed paths}, and their distributed computation.

For a non-leaf vertex $v \in V(T)$, its \emph{heavy child}, denoted $v_h$, is the (unique) child $v'$ of $v$ maximizing\footnote{Ties are broken arbitrarily and consistently.} the number of vertices in its subtree $T_{v'}$. Any other child of $v$ is a \emph{light child}.
A vertex is \emph{heavy} if it is a heavy child, and \emph{light} otherwise (so the root $s$ is light).
A tree edge is \emph{heavy} if it connects a vertex to its heavy child, and \emph{light} otherwise.
If $(u,u')$ is a heavy (resp., light) edge in the path $\pi(s,v)$, then $u$ is a \emph{heavy ancestor} (resp., \emph{light ancestor}) of $v$, and $v$ is a \emph{heavy descendant} (resp, \emph{light descendant}) of $u$.
(Note that in our terminology, a `heavy ancestor' need not be a heavy vertex itself, and similarly for all combinations of heavy/light and ancestor/descendant).
We denote by $\LA(v)$ (resp., $\LD(v)$) the set of $v$'s light ancestors (resp., descendants).
It is easy to show that $\pi(s,v)$ contains $O(\log n)$ \emph{light} vertices and edges, hence also $\vert \LA(v) \vert = O(\log n)$

\begin{definition}[Compressed paths]
Let $v \in V(T)$. Let $L = [s = v_0, v_1, \dots, v_k]$ be the ordered list of the light vertices on the root-to-$v$ path $\pi(s, v, T)$. The \emph{compressed path} of $v$ with respect to $T$, denoted $\pi^*(s, v, T)$ consists of the list $L$, along with a table mapping each $v_i$ to the number of heavy vertices appearing between $v_i$ and $v_{i+1}$ in $\pi(s, v, T)$ (where we define $v_{k+1} = v$). Note that the compressed path $\pi^*(s,v,T)$ has bit-length $O(\log^2 n)$.
\end{definition}

Observe that compressed paths can be used as \emph{ancestry labels} in $T$: given $\pi^*(s,u,T)$ and $\pi^*(s,v,T)$, one can check whether $\pi(s,u,T)$ is a prefix of $\pi(s,v,T)$, and hence determine whether $u$ is an ancestor of $v$ in $T$.

\begin{lemma}\label{lem:distributed_heavy_light}
	For every tree $T$, there is an $\widetilde{O}(D(T))$-rounds $\widetilde{O}(1)$-congestion algorithm letting each vertex $v$ of $T$ learn its heavy/light classification and its compressed path $\pi^* (s,v, T)$. 
\end{lemma}
\def\APPENDHEAVYLIGHTDIST{
    % \begin{proof}[Proof of Lemma \ref{lem:distributed_heavy_light}]
    \begin{proof}
		First, each vertex $v$ learns its subtree size $\vert V(T_v)\vert $ by bottom-up aggregation on $T$. By passing these sizes to the parents, within another round each vertex can classify its children as either heavy or light. Within one more round, each vertex is informed of its classification by its parent. Computing the compressed paths can now be executed in a top-down fashion, as a vertex can deduce its own compressed path from the compressed path of its parent and its own heavy/light classification.
	\end{proof}
}\APPENDHEAVYLIGHTDIST

\paragraph{Graph Sketches.}
We give a formal but brief definition of graph sketches. We follow \cite{DoryP21}, and refer the reader to Section 3.2.1 therein for a detailed presentation of the subject.
Throughout, let $\oplus$ denote the bitwise-XOR operator.
The first required ingredients are randomized \emph{unique edge identifiers}:
\begin{lemma}[{\cite[Lemma 3.8]{DoryP21}}, {\cite[Lemma 2.4]{GhaffariP16}}]\label{lem:uids}
    Using a random seed $\mathcal{S}_{\ID}$ of $O(\log^2 n)$ random bits, one can compute a collection of $M=\binom{n}{2}$ $O(\log n)$-bit identifiers for the pairs in $\binom{V}{2}$, denoted $\mathcal{I}=\{\UID(e_1), \ldots, \UID(e_{M})\}$, with the following property: For any nonempty subset $E' \subseteq E$ with $\vert E' \vert \neq 1$, $\Pr[\oplus_{e \in E'} \UID(e) \in \mathcal{I}] \leq 1/n^{10}$. Furthermore, for any $e = (u,v)$, the identifier $\UID(e)$ can be computed from $\ID(u)$, $\ID(v)$ and the random seed $\mathcal{S}_{\ID}$.
\end{lemma}

Next, we define the notion of \emph{extended edge identifiers}, formed by augmenting $\UID(e)$ with the IDs and $T$-ancestry labels of the endpoints based on compressed paths, namely $\LCALabel_T(v) = \pi^* (s, v, T)$. Formally, an edge $e =(u,v)$ we have
%\begin{equation}\label{eq:extend-ID}
%	\EID_T(e)=[\UID(e), \ID(u), \ID(v), \LCALabel_T(u), \LCALabel_T(v)]~.
%\end{equation}
\begin{gather}\label{eq:extend-ID}
\EID_T(e)=
[\UID(e), \ID(u), \ID(v), \LCALabel_T(u), \LCALabel_T(v)]~.
\end{gather}

We are ready to define the sketches. We now follow \cite{DuanConnectivityArxiv16,DuanConnectivitySODA17, DoryP21} and use pairwise independent hash functions for this purpose. Choose $L = c \log n$ pairwise independent hash functions $h_1, \ldots, h_{L}:\{0,1\}^{\Theta(\log n)} \to \{0, \ldots, 2^{\log M}-1\}$, and for each $i \in [1, L]$ and $j \in [0,\log M]$ define the edge set 
$E_{i,j} =\{ e \in E \mid h_i(e) \in [0,2^{\log M-j})\}$.
Each $h_i$ can be defined by a random seed of logarithmic length \cite{TCS-010}. Thus, a random seed $\mathcal{S}_h$ of length $O(L \log n)$ can be used to determine the collection of all $L$ hash functions.
For each vertex $v$ and indices $i,j$, let $E_{i,j}(v)$ be the edges incident to $v$ in $E_{i,j}$. 
The $i^{th}$ \emph{basic sketch unit} of a vertex $v$ is then given by:
%$$\Sketch_{G,i}(v)=[\oplus_{e \in E_{i,0}(v)} \EID_T (e),\ldots, \oplus_{e \in E_{i, \log M}(v)} \EID_T (e)].$$
\begin{gather*}
	\Sketch_{G,i}(v)= 
	\big[\bigoplus_{e \in E_{i,0}(v)} \EID_T (e),\ldots, \bigoplus_{e \in E_{i, \log M}(v)} \EID_T (e)\big].
\end{gather*}
We extend this definition to vertex subsets by XORing, i.e., for $S \subseteq V$, 
$\Sketch_{G, i}(S)=\oplus_{v \in S}\Sketch_{G,i}(v).$
The \emph{sketch} of a vertex $v$ is formed by  concatenating its basic sketch units: 
%$$\Sketch_G(v)=[\Sketch_{G,1}(v),\Sketch_{G,2}(v), \ldots\Sketch_{G,L}(v)]~.$$ 
\begin{gather*}
	\Sketch_G(v)= 
	[\Sketch_{G,1}(v),\Sketch_{G,2}(v), \ldots\Sketch_{G,L}(v)]~.
\end{gather*}
Again, we extend this definition to $S \subseteq V$ by
$\Sketch_G(S)=\bigoplus_{v \in S}\Sketch_G(v)$. 

\begin{observation}\label{obs:0-sketch}
    $\Sketch_{G} (V) = 0$, where $0$ denote the all-$0$s string of appropriate length.
\end{observation}
\begin{proof}
    For each $i,j$, the $j^{th}$ entry of $\Sketch_{G,i} (V)$ contains 
    $
    \bigoplus_{v \in V} \bigoplus_{e \in E_{i,j}(v)} \EID(e) ~.
    $
    Each term in this XOR appears precisely twice: $\EID(e)$ of $e = (u,w) \in E_{i,j}$ appears once for $u$ and once for $w$. 
    Hence, this entry equals $0$, and the result follows.
\end{proof}

The main use of graph sketches is in finding outgoing edges:
\begin{lemma}[{\cite[Lemma 3.11]{DoryP21}}]\label{lem:sketch-property}
	For any subset $S$, given a basic sketch unit $\Sketch_{G,i}(S)$ and the seed $\mathcal{S}_{\ID}$ one can compute, with constant probability\footnote{Over the choice of the random seeds $\mathcal{S}_{\ID}$ and $\mathcal{S}_{h}$.} $\EID_T (e)$ for an outgoing edge $e$ from $S$ in $G$, if such exists.
\end{lemma}

The following lemma lets us ``cancel" edges from the sketches:

\begin{lemma}\label{lem:cancel-sketch-property}
	Let $S \subseteq V$, and let $E' \subseteq E$ be a set of outgoing edges from $S$. Then, given  $\Sketch_{G} (S)$, the random seeds $\mathcal{S}_h, \mathcal{S}_{\ID}$, and the extended identifiers $\EID_T(e)$ of all $e \in E'$, one can compute a \emph{cancellation-sketch} $\CanSketch(E')$ for $E'$ such that $\Sketch_{G \setminus E'} (S) = \Sketch_{G}(S) \oplus \CanSketch(E')$.
\end{lemma}
\def\APPENDCANCELSKETCH{
\begin{proof}
% \begin{proof}[Proof of Lemma \ref{lem:cancel-sketch-property}]
	Let $Out(S)$ be the set of outgoing edges from $S$ in $G$, and $Out_{i,j}(S) = Out(S) \cap E_{i,j}$. Observe that for each $i \in \{1, \dots, L\}$:
%	\begin{align*}
%		\Sketch_{G,i}(S)
%		&=
%		[\oplus_{v \in S} \oplus_{e \in E_{i,0}(v)} \EID_T (e),\ldots, \oplus_{v \in S} \oplus_{e \in E_{i,\log M}(v)} \EID_T (e) ] \\
%		& = [\oplus_{e \in Out_{i,0}(S)} \EID_T (e),\ldots, \oplus_{e \in Out_{i, \log M}(S)} \EID_T (e)]
%	\end{align*}
\begin{align*}
    \Sketch_{G,i}(S)
    &= 
    [\bigoplus_{v \in S} \bigoplus_{e \in E_{i,0}(v)} \EID_T (e),\ldots, \bigoplus_{v \in S} \bigoplus_{e \in E_{i,\log M}(v)} \EID_T (e) ] \\
    &=
    [\bigoplus_{e \in Out_{i,0}(S)} \EID_T (e),\ldots, \bigoplus_{e \in Out_{i, \log M}(S)} \EID_T (e)]
\end{align*}
	where the last equality is true as each edge $e$ with both endpoints inside $S$ appears either $0$ or $2$ times in each XOR. Now let $Out'(S)$ and $Out'_{i,j}(S)$ be defined exactly as $Out(S)$ and $Out_{i,j}(S)$, but with respect to $G \setminus E'$ instead of $G$. Then as $Out(S)$ is the disjoint union of $Out'(S)$ and $E'$, we obtain
%$$\Sketch_{G,i}(S) \oplus \Sketch_{G \setminus E',i}(S) = [\oplus_{e \in E' \cap E_{i, 0}} \EID_T (e),\ldots, \oplus_{e \in E' \cap E_{i, \log M}} \EID_T (e)].$$
	\begin{gather*}
		\Sketch_{G,i}(S) \oplus \Sketch_{G \setminus E',i}(S) = 
		[\bigoplus_{e \in E' \cap E_{i, 0}} \EID_T (e),\ldots, \bigoplus_{e \in E' \cap E_{i, \log M}} \EID_T (e)].
	\end{gather*}
	The right-hand side of the above equation can be computed from the given extended IDs of $E'$ and the random seeds $\mathcal{S}_h, \mathcal{S}_{\ID}$, and we denote it by $\CanSketch_{i}(E')$. Note that by the above equation, we have that $\Sketch_{G \setminus E',i}(S) = \Sketch_{G,i}(S) \oplus \CanSketch_{i}(E')$. Therefore, we can define the cancellation-sketch as the concatenation $\CanSketch(E') = [\CanSketch_{1}(E'), \ldots, \CanSketch_{L}(E')]$, and the lemma follows.
\end{proof}
}\APPENDCANCELSKETCH

\paragraph{Distributed Scheduling.}
The congestion of an algorithm $\cA$ is defined by the worst-case upper bound on the number of messages exchanged through a given graph edge when simulating $\cA$.  Throughout, we make extensive use of the following random delay approach of \cite{leighton1994packet}, adapted to the \congest\ model. 
\begin{theorem}[{\cite[Theorem 1.3]{Ghaffari15}}]\label{thm:delay}
	Let $G$ be a graph and let $\cA_1,\ldots,\cA_m$ be $m$ distributed algorithms, each algorithm takes at most $\dilation$ rounds, and where for each edge of $G$, at most $\congestion$ messages need to go through it, in total 
	\emph{over all} these algorithms. Then, there is a randomized distributed
	algorithm that w.h.p.\ runs all the algorithms in $\widetilde{O}(\congestion +\dilation)$ rounds.
\end{theorem}

\section{Single Cut Vertices} \label{sec:single-cut}

In this section we describe the distributed algorithm for detecting single vertex cuts of Theorem \ref{thm:distributed-single-cut}.
This serves both as a warm-up to our approach in the subsequent sections devoted to cut pairs detection, as well as for a detailed presentation of basic tools used in these next sections. 
We assume each vertex $v$ is equipped with its heavy/light classification in $T$ (recall that $T$ is a BFS tree of $G$ rooted at $s$), and with its ancestry label which is its compressed path, $\LCALabel_T (v) = \pi^* (s, v, T)$. This can be achieved in $\widetilde{O}(D)$ rounds by Lemma \ref{lem:distributed_heavy_light}.

\paragraph{Step 1: Computing Subtree Sketches.}
The source $s$ samples the random seeds $\mathcal{S}_{\ID}, \mathcal{S}_h$ of $\widetilde{O}(1)$ bits, and broadcasts them to all vertices. Using Lemma \ref{lem:uids}, each vertex $v$ can then locally compute the $\UID(e)$ for each edge $e$ incident to $v$. By letting all neighbors in $G$ exchange their $\LCALabel_T$-labels, each $\UID(e)$ can be concatenated with the required information to create $\EID(e)$. This provides each vertex $v$ with all the information needed to lcoally compute $\Sketch_{G} (v)$. By XOR-aggregation of the individual sketches from the leaves of $T$ upwards, each vertex $v$ obtains its subtree sketch, given by $\Sketch_G(V(T_v))=\oplus_{u \in T_v} \Sketch_{G} (v)$.
Next, each vertex passes its subtree sketch to its parent, so that each vertex now holds the subtree sketch for each of its children. Finally, the source $s$ also broadcasts its subtree sketch, which is $\Sketch_G(V)$, to all the other vertices.
% The $\widetilde{O}(D)$-rounds, $\widetilde{O}(1)$-congestion bounds are straightforward.
The $\widetilde{O}(D)$-rounds and $\widetilde{O}(1)$-congestion bounds follow as we have performed a constant number of broadcasts and aggregations of $\widetilde{O}(1)$-bit strings over $T$, which has diameter $O(D)$ as a BFS tree of $G$ \cite{Peleg:2000}.

\paragraph{Step 2: Local Bor\r{u}vka Simulation.} This step is locally applied at every vertex $x$, and requires no additional communication rounds. We show that given the information of Step 1, $x$ can locally simulate the Bor\r{u}vka's algorithm \cite{Boruvka} in the graph $G \setminus \{x\}$, and thus determine if $G \setminus \{x\}$ is connected.
Let $x_1, \ldots, x_{k}$ be the children of $x$ in $T$. We assume that $x \neq s$; the case $x = s$ is easier and requires only slight modifications. The connected components in $T \setminus \{x\}$ are denoted by 
$
\mathcal{P}_x=\{ V(T_{x_j}) \mid j=1, \dots, k \} \cup \{V \setminus V(T_x)\}.
$
By Step 1, $x$ holds the $G$-sketch of each component in $\mathcal{P}_x$: It has explicitly received $\Sketch_G(V(T_{x_j}))$ from each child $x_j$. The sketch of the remaining component can be locally inferred as $\Sketch_G (V \setminus V(T_x))=\Sketch_G (V) \oplus \Sketch_G (V(T_x))$. 
For a neighbor $u$ of $v$, $x$ can use $\LCALabel_T (u)$ (found in $\EID_T ((x,u))$) to determine component containing $u$ in $\mathcal{P}_x$. Therefore, using \Cref{lem:cancel-sketch-property}, $x$ can cancel its incident edges to obtain $\Sketch_{G \setminus \{x\}} (P)$ for every $P \in \mathcal{P}_x$.

We are now ready to describe the Bor\r{u}vka execution, which is very similar to the (centralized) decoding algorithm of \cite{DoryP21}. The algorithm consists of $O(\log n)$ phases. Each phase $i$ is given as input a partitioning $\mathcal{P}_{x,i} = \{P_{i,1}, \ldots, P_{i,k_i}\}$ of $V \setminus \{x\}$ into connected \emph{parts}, along with their sketch information $\Sketch_{G \setminus \{x\}} (P_{i,j})$. The initial partitioning given to the first phase is $\mathcal{P}_{x,0} = \mathcal{P}_x$. The output of the phase is a coarser partitioning  $\mathcal{P}_{x,i+1}$ along with the $(G\setminus \{x\})$-sketch information of the new parts. A part $P_{i,j} \in \mathcal{P}_{x,i}$ is said to be \emph{growable} if it has at least one outgoing edge to a vertex in $V \setminus (P_{i,j} \cup \{x\})$. To obtain outgoings edges from the growable parts in $\mathcal{P}_{x,i}$, the algorithm uses the $i^{th}$ basic-unit sketch $\Sketch_{G \setminus \{x\}, i}(P_{i,j})$ of each $P_{i,j} \in \mathcal{P}_{x,i}$. By Lemma \ref{lem:sketch-property}, from every growable part $P_{i,j} \in \mathcal{P}_{x,i}$, we get $\EID_T (e)$ for one outgoing edge $e=(u,v)$ with constant probability.
To find the part $P_{i,j'}$ containing the other endpoint of $e$ (to be merged with $P_{i,j}$), we use the $T$-ancestry labels found in $\EID_T(e)$. Say this endpoint is $v$. We determine the component of $v$ in $T \setminus \{x\}$, i.e.\ the part $P_{0,q}$ containing $v$ in $\mathcal{P}_{x,0}$, by querying the ancestry relation between $v$ and the children of $x$ using their $\LCALabel_T$-labels. Then $v$ belongs to the unique component $P_{i,j'} \in \mathcal{P}_{x,i}$ containing $P_{0,q}$. The sketch information for the next phase $i+1$ is given by XORing over the sketches of the parts in $\mathcal{P}_{x,i}$ that got merged into a single component in $\mathcal{P}_{x,i+1}$.

Note that it is important to use fresh randomness (i.e.\ independent sketch information) in each of the Bor\r{u}vka phases \cite{ahn2012analyzing,kapron2013dynamic,DuanConnectivityArxiv16}. Since each growable component gets merged with constant probability, the expected number of growable components is reduced by a constant factor in each phase. Thus after $O(\log n)$ phases, the expected number of growable components is at most $1/n^5$, and by Markov's inequality we conclude that w.h.p.\ there are no growable components. 
% The partitioning at this point corresponds to the maximal connected components in $G \setminus \{x\}$, so its connectivity can be inferred. 
Namely, at this point, the partition is just the connected components of $G \setminus \{x\}$ (as a part is not growable iff it has no outgoing edges in $G\setminus\{x\}$, i.e.\ it is a connected component of this graph).
We therefore determine that $G \setminus \{x\}$ is disconnected (i.e.\ that $x$ is a cut vertex) iff this partition has more than one part.

This concludes the proof of Theorem \ref{thm:distributed-single-cut}.
Finally, we note that by tracking the merges throughout the Bor\r{u}vka simulation, $x$ can also find a subset $\widetilde{E}$ of the outgoing edges received throughout the simulation such $(T \setminus \{x\}) \cup \widetilde{E}$ is a maximal spanning forest of $G \setminus \{x\}$. This becomes useful in next sections.

\section{Dependent Cut Pairs} \label{sec:depend}
In this section we present an $\widetilde{O}(D)$-round algorithm for detecting \emph{dependent} cut pairs in $G$, i.e.\ pairs $xy$ where $x$ is a descendant of $y$ in the BFS tree $T$ rooted at $s$. Recall that our approach is based on scheduling the execution of algorithms $\{\cA_y\}_{y \in V}$, where $\cA_y$ detects all cut pairs $xy$ such that $x \in T_y$ (see overview in Section \ref{sec:approach}). By employing the single cut vertices detection algorithm of Section \ref{sec:single-cut} as a common preprocessing phase prior to the execution of the $\{\cA_y\}_{y \in V}$ algorithms, we may assume that there are no single cut vertices in $G$. Furthermore, by carefully examining the properties of this algorithm, we may assume that every $v \in V$ holds the following preprocessing information:

\begin{itemize}
	\item The random seeds $\mathcal{S}_{\ID},\mathcal{S}_h$.
	\item $\EID_T (e)$ for every edge $e$ incident to $v$.
	\item $\vert V(T_v)\vert $ and $\vert V(T_{v'})\vert $ for every $T$-child $v'$ of $v$.
	\item $\Sketch_{G} (v)$, $\Sketch_{G} (V)$, $\Sketch_{G} (V(T_v))$ and $\Sketch_{G} (V(T_{v'}))$ for every $T$-child $v'$ of $v$.
	\item An edge set $\widetilde{E}(v) \subseteq E \setminus E(T)$ such that $\widetilde{T}(v) = (T \setminus \{v\}) \cup  \widetilde{E}(v)$ is a spanning tree of $G \setminus \{v\}$.
	For each $e \in  \widetilde{E}(v)$, its extended identifier $\EID_T (e)$ is known.
\end{itemize}

\begin{lemma} \label{lem:Ay}
	Assuming all vertices know their preprocessing information, there is an $\widetilde{O}(D)$-rounds, $\widetilde{O}(1)$-congestion algorithm $\cA_y$ that detects all cut pairs $xy$ where $x \in V(T_y)$. The algorithm $\cA_y$ sends messages only on edges incident to $V(T_y)$.
\end{lemma}

We next describe the algorithm $\cA_y$.
Throughout, let $\widetilde{E} = \widetilde{E}(y)$, $\widetilde{T} = \widetilde{T}(y)$, $\widetilde{G} = G \setminus \{y\}$, and denote the $T$-children of $y$ by $y_1, \dots, y_k$.

\paragraph{Step 0: Local Computation of Component Tree for $\widetilde{T}$ in $y$.} 
This preliminary step is executed by local computation in $y$. It constructs the \emph{component tree} $\widetilde{CT}$, which is obtained from $\widetilde{T}$ by contracting each connected component of $\widetilde{T} \setminus \widetilde{E}$ into a single node. Note that $\widetilde{T} \setminus \widetilde{E} = T \setminus \{y\}$, namely the nodes in $\widetilde{CT}$ correspond to connected components of $T \setminus \{y\}$. More concretely, for every $i = 1, \dots, k$ the component $C_i = V(T_{y_i})$ is a node of $\widetilde{CT}$, and (unless $y = s$) there is another node for the component $C_0 = V \setminus V(T_y)$.
Each edge $(C_i,C_j)$ in $\widetilde{CT}$ corresponds to the unique $\widetilde{E}$-edge incident to both $C_i$ and $C_j$. Observe that the extended edge identifiers which are known to $y$ by the preprocessing contain the $T$-ancestry labels of all the endpoints of the edges in $\widetilde{E}$, as well as those of the $y_i$'s. Using these ancestry labels, $y$ can determine the components incident to each edge $e \in \widetilde{E}$, and therefore construct $\widetilde{CT}$ locally.

For clarity of presentation we assume $y \neq s$; the special case $y=s$ is easier, and requires only slight modifications. We set $s$ as the root of $\widetilde{T}$, and accordingly $C_0$ is the root of $\widetilde{CT}$. For each $i = 1, \dots , k$, denote by $e_i = (r_i, p_i)$ the unique edge in $\widetilde{E}$ connecting $C_i$ to its parent in $\widetilde{CT}$, where $r_i$ is the endpoint of $e_i$ inside $C_i$, and $p_i$ is the endpoint lying in the parent component. See \Cref{fig:component-tree} for an illustration.

\begin{figure}
	\begin{center}
		\includegraphics[height=7cm]{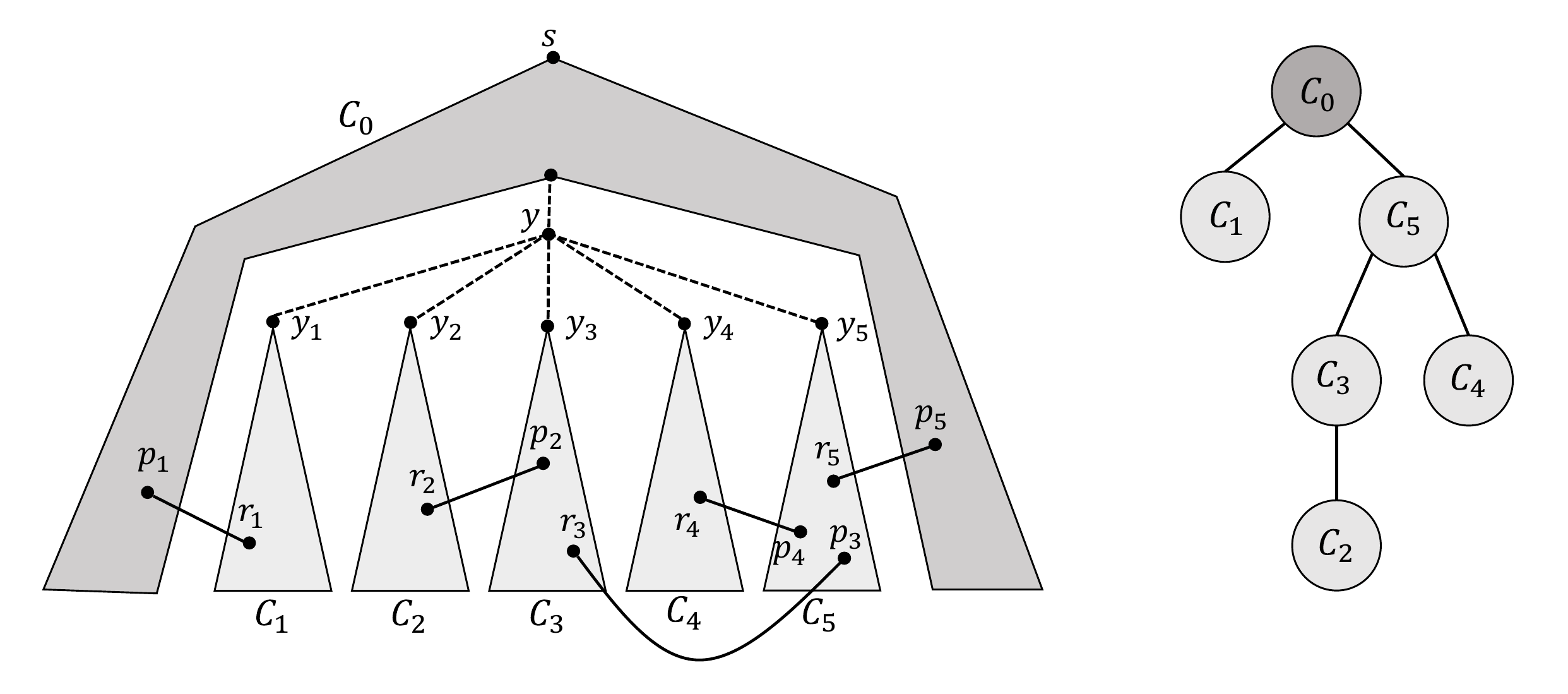}
		\caption{ \sf Left: Illustration of the trees $T$ and $\widetilde{T}$. The dashed edges are $T$-edges incident to $y$, and the solid edges are $\widetilde{E}$-edges. The components $C_0, C_1, \dots, C_5$ are each internally connected via original $T$-edges. The tree $\widetilde{T}$ is obtained by removing $y$ and its incident edges from $T$, and adding the $\widetilde{E}$ edges.
		Right: The component tree $\widetilde{CT}$. 
			 \label{fig:component-tree} 
		}
	\end{center}
\end{figure}

\paragraph{Step 1: Construction of $\widetilde{T}$.}
The goal of this step is for each vertex in $V(T_y) \setminus \{y\}$ to learn its parent in $\widetilde{T}$. First, $y$ sends its children their corresponding edges from $\widetilde{E}$, so each $y_i$ learns $\EID_T (e_i = (r_i, p_i))$. The $y_i$'s propagate (in parallel) these edges down their $T$-subtrees, so that all vertices of component $C_i$ learn $\EID_T(e_i)$. Then, a BFS procedure with source vertex $r_i$ is executed inside each tree $T_{y_i}$ (in parallel). This completes the step, since the $\widetilde{T}$-parent of each vertex in $C_i$ is its BFS-parent from this last procedure, except for $r_i$ whose $\widetilde{T}$-parent is $p_i$.

\paragraph{Step 2: Computing $\widetilde{T}$-Ancestry labels.}
In later steps, we locally simulate Bor\r{u}vka's algorithm similarly to Section \ref{sec:single-cut}, but with the initial components being parts of $\widetilde{T}$. In order to identify which components get merged by the outgoing edges, we need ancestry labels with respect to $\widetilde{T}$ rather than $T$.
As we are restricted to send messages only on $V(T_y)$-incident edges, we would like the $T$- and $\widetilde{T}$-labels to coincide for vertices in $C_0$ (as some of them cannot be informed of new labels). Note that the compressed paths of $v \in C_0$ w.r.t.\ $T$ and $\widetilde{T}$ are generally different, even though $\pi(s,v,T) = \pi(s,v,\widetilde{T})$, as the these trees have different heavy-light notions. Hence, instead of relying solely on compressed paths in $\widetilde{T}$, we take a hybrid approach and define new labels based on breaking each $\widetilde{T}$-path to a $T$-part and a strictly $\widetilde{T}$-part, and compressing them accordingly.
We still have the challenge of computing (at least part of) the heavy-light decomposition of $\widetilde{T}$. As the diameter of $\widetilde{T}$ might be $\Omega(\Delta D)$, we cannot use simple bottom-up or top-down computations on $\widetilde{T}$. The key to overcoming this is utilizing $y$ as a coordinator, enabling the parts $C_i$ to work in parallel.
The full details appear in the proof of the next claim, which is deferred to \Cref{sec:depend-imp-details} on implementation details.
% The full details appear in the proof of the next claim. Missing proofs in this section are found in Appendix \ref{app:depend}.

\begin{claim}\label{cl:new-ancestry-labels}
	In $\widetilde{O}(D)$-rounds of computation with $\widetilde{O}(1)$ congestion, in which messages are sent only on $V(T_y)$-incident edges, one can compute $\widetilde{T}$-ancestry labels $\LCALabel_{\widetilde{T}} (\cdot)$ of $\widetilde{O}(1)$ bits, such that \emph{every} vertex $v$ of $\widetilde{T}$ learns $\LCALabel_{\widetilde{T}} (v)$.
\end{claim}
\def\APPENDNEWANCESTRYLABELS{
\begin{proof}[Proof of Claim \ref{cl:new-ancestry-labels}]
	We define the labels $\LCALabel_{\widetilde{T}}$ as follows. If $v \in C_0$, we simply take its $T$-ancestry label, i.e.\ $\LCALabel_{\widetilde{T}} (v) = \LCALabel_T (v) = \pi^* (s, v, T)$. We now define $\LCALabel_{\widetilde{T}} (v)$ for $v \in C_i$, $i \neq 0$. The $s$-$v$ $\widetilde{T}$-path decomposes as $\pi(s, v, \widetilde{T}) = \pi(s,p_j,T) \circ (p_j, r_j) \circ \pi(r_j, v, \widetilde{T}_{r_j})$ for the unique $p_j \in C_0$ which is an ancestor of $v$ in $\widetilde{T}$. Then $\LCALabel_{\widetilde{T}} (v)$ is obtained by heavy-light compression of these paths as $\LCALabel_{\widetilde{T}} (v) = \pi^* (s, p_j, T) \circ (p_j, r_j) \circ \pi^* (r_j, v, \widetilde{T}_{r_j})$. (The heavy-light notions in the first segment are those of $T$, while in the last segment they are those of $\widetilde{T}$.)
	
	Given the $\LCALabel_{\widetilde{T}}$-labels of any $u,v \in V(\widetilde{T})$, one can easily determine if $\pi(s, u, \widetilde{T})$ is a prefix of $\pi(s, v, \widetilde{T})$, and thus if $u$ is an $\widetilde{T}$-ancestor of $v$. As compressed paths require $O(\log^2 n)$ bits, each $\LCALabel_{\widetilde{T}}$-label consists of $O(\log^2 n) = \widetilde{O}(1)$ bits.
	
	We now compute the $\LCALabel_{\widetilde{T}}$-labels. The vertices of $C_0$ already hold them, as they are equal to their $\LCALabel_T$-labels. It remains to compute them for the vertices of $C_1 \cup \cdots \cup C_k$. This is done in three steps, as follows.
	
	\textbf{Step A: Heavy-Light Decomposition of $\widetilde{T}$.}
	Our first task is letting each $v \in C_1 \cup \cdots \cup C_k$ learn its heavy/light classification in $\widetilde{T}$. This essentially involves computing subtree sizes in $\widetilde{T}$. At first glance, one might consider computing these by simple aggregation on $\widetilde{T}$. This approach fails, the height of $\widetilde{T}$ might be $\Omega(\Delta \cdot D)$. To overcome this, we use $y$ as a coordinator to jump-start the aggregation. Recall that all component sizes $\vert C_i\vert $ are known to $y$ by the preprocessing. Therefore, $y$ can locally compute the $\widetilde{T}$-subtree sizes of all the $r_i$'s: $\vert V(\widetilde{T}_{r_i})\vert $ is the sum-of-sizes of components lying in the subtree of $C_i$ in the component tree $\widetilde{CT}$. Then, $y$ sends $\vert V(\widetilde{T}_{r_i})\vert $ to each child $y_i$, and this is propagated down $T_{y_i}$ (in parallel), so that each $r_i$ learns its $\widetilde{T}$-subtree size. Each $r_i$ passes this information also to its $\widetilde{T}$-parent $p_i$. For any vertex $v \in V \setminus \{y\}$, define
	\[
	\alpha_v =
	\begin{cases}
		\text{if $v = p_i$:} & \vert V(\widetilde{T}_{r_i})\vert  + 1, \\
		\text{otherwise:} & 1.
	\end{cases}
	\]
	Then by this point, every vertex $v \in C_1 \cup \cdots \cup C_k$ knows its corresponding value $\alpha_v$. For $i = 1, \dots, k$, let $\widetilde{T}^{(i)}$ be the tree induced on $C_i$ by $\widetilde{T}$, where the parents in $\widetilde{T}^{(i)}$ are the same as in $\widetilde{T}$. Equivalently, $\widetilde{T}^{(i)}$ is the tree obtained by rerooting $T_{y_i}$ at the vertex $r_i$. Each of its leaves is either an original $\widetilde{T}$-leaf or a $p_j$ vertex for some $j$. The crux is that for each $v \in C_i$ it holds that $\vert V(\widetilde{T}_{v})\vert  = \sum_{u \in \widetilde{T}^{(i)}_v} \alpha_u$. That is, the $\widetilde{T}$-subtree size of $v$ is equal to the sum-of-$\alpha$'s in its $\widetilde{T}^{(i)}$-subtree. By executing bottom-up sum-aggregation of the $\alpha_v$'s in each of the trees $\widetilde{T}^{(i)}$ in parallel, each $v \in C_1 \cup \cdots \cup C_k$ learns its $\widetilde{T}$-subtree size. Each such vertex $v$ then passes its $\widetilde{T}$-subtree size to its parent, enabling the parents to classify their children into light or heavy in $\widetilde{T}$, and inform them of their classification.
	
	\textbf{Step B: Computing Compressed $\widetilde{T}$-Paths Inside the $C_i$'s}.
	In this step, each $v \in C_i$, $i \neq 0$, learns the compressed path $\pi^* (r_i, v, \widetilde{T}_{r_i})$. The main observation is that if a vertex is given the compressed path of its parent, it can easily deduce its own compressed path (as it know its own heavy/light classification). Therefore, the required compressed paths can be computed in a top-down fashion on each $\widetilde{T}^{(i)}$ (in parallel).
	
	\textbf{Step C: Obtaining The $\LCALabel_{\widetilde{T}}$-labels.}
	For $j = 1, \dots, k$ define
	\[
	\pi_j = 
	\begin{cases}
		\text{if $p_j \in C_i$ with $i \neq 0$:} & \pi^*(r_i, p_j, \widetilde{T}_{r_i}), \\
		\text{if $p_j \in C_0$:} & \pi^*(s, p_j, T) = \LCALabel_T (p_j).
	\end{cases}
	\]
	Observe that by Step B, the $p_j$'s know their corresponding $\pi_j$'s. To send this information to $y$, each $p_j$ sends $\pi_j$ to $r_j$, and the messages are then forwarded upwards on each $T_{y_j}$ (in parallel), along with the heavy/light classification of the $r_j$'s. Using the information of $\{\pi_j\}_j$, the component tree $\widetilde{CT}$ and the heavy/light classifications of the $r_j$'s, $y$ can locally compute $\LCALabel_{\widetilde{T}} (r_i)$ for each $r_i$, and send this label to the corresponding child $y_i$. Then, $\LCALabel_{\widetilde{T}} (r_i)$ is broadcasted on each $T_{y_i}$. Finally, $\LCALabel_{\widetilde{T}} (v)$ can be locally deduced in $v$ from the information in $\LCALabel_{\widetilde{T}} (r_i)$ and $\pi^* (v, \widetilde{T}_{r_i})$, where the latter is known to $v$ by Step B.
\end{proof}
}%\APPENDNEWANCESTRYLABELS

\paragraph{Step 3: Computing Sketches w.r.t.\ $\widetilde{G}$ and $\widetilde{T}$.}
We define new extended edge identifiers for edges of $\widetilde{G} = G \setminus \{y\}$, based on its spanning tree $\widetilde{T}$. For an edge $e = (u,v)$ of $G \setminus \{y\}$, let
%$$\EID_{\widetilde{T}}(e) =[\UID(e), \ID(u), \ID(v), \LCALabel_{\widetilde{T}}(u), \LCALabel_{\widetilde{T}}(v)].$$
\begin{gather*}
	\EID_{\widetilde{T}}(e) = 
	[\UID(e), \ID(u), \ID(v), \LCALabel_{\widetilde{T}}(u), \LCALabel_{\widetilde{T}}(v)].
\end{gather*}
Next, for every vertex $v \in V \setminus \{y\}$ we define $\Sketch_{\widetilde{G}} (v)$ exactly as $\Sketch_{G} (v)$, only ignoring edges incident to $y$ in the sampling, and using the $\EID_{\widetilde{T}}$ identifiers for the edges.
Computing these new sketches requires $\widetilde{O}(1)$ rounds of communication, in which every $v \in C_1 \cup \cdots \cup C_k$ sends $\LCALabel_{\widetilde{T}}(v)$ to all its $\widetilde{G}$-neighbors. As the $T$- and $\widetilde{T}$-ancestry labels coincide on the vertices of $C_0$, every vertex $v \in V \setminus \{y\}$ can now determine $\EID_{\widetilde{T}} (e)$ for every $\widetilde{G}$-edge $e$ incident to $v$, and use the random seed $\mathcal{S}_h$ to compute $\Sketch_{\widetilde{G}} (v)$.

% \textbf{3.1: Computing $\widetilde{T}$-Subtree Sketches.}
Our next goal is letting each vertex $x \in C_1 \cup \cdots \cup C_k$ learn the $\widetilde{G}$-sketch of its $\widetilde{T}$-subtree (not $T$-subtree), namely $\Sketch_{\widetilde{G}} (V(\widetilde{T}_x)) = \bigoplus_{v \in \widetilde{T}_x} \Sketch_{\widetilde{G}} (v)$.
This is done by using $y$ as a coordinator, similarly to the $\widetilde{T}$-subtree sum computation in the proof of \Cref{cl:new-ancestry-labels}.
We start by XOR-aggregation of the $\widetilde{G}$-sketches upwards on each $T_{y_i}$ (in parallel), which takes $\widetilde{O}(D)$ rounds, and results in each $y_i$ learning  $\Sketch_{\widetilde{G}} (C_i)$. Within $\widetilde{O}(1)$ rounds, these component $\widetilde{G}$-sketches are passed to $y$ from its children. $y$ can now locally compute the $\widetilde{T}$-subtree $\widetilde{G}$-sketch of each $r_i$ as follows:
$
\Sketch_{\widetilde{G}} (V(\widetilde{T}_{r_i})) = \oplus_{j \in J(i)} \Sketch_{\widetilde{G}} (C_j)
$,
where $J(i)$ is the set of all indices $j$ such that $C_j$ is the subtree of $C_i$ in the component tree $\widetilde{CT}$. Then, $y$ sends each of its children $y_i$ the $\widetilde{T}$-subtree sketch of $r_i$, and which is propagated down on each $T_{y_i}$ (in parallel) so that each $r_i$ learns $\Sketch_{\widetilde{G}} (V(\widetilde{T}_{r_i}))$. The $r_i$'s then send this information to their $\widetilde{T}$-parent, which are the $p_i$'s. For each vertex $v$ of $\widetilde{T}$, let
$$
\beta_v =
\begin{cases}
	\text{if $v = p_j$:} & \Sketch_{\widetilde{G}} (v) + \Sketch_{\widetilde{G}} (V(\widetilde{T}_{r_j})) \\
	\text{otherwise:} & \Sketch_{\widetilde{G}} (v)
\end{cases}
$$
Then by this point, every $v \in C_1 \cup \cdots \cup C_k$ know its $\beta_v$ value.
For $i = 1, \dots, k$, let $\widetilde{T}^{(i)}$ be the tree induced on $C_i$ by $\widetilde{T}$, where the parents in $\widetilde{T}^{(i)}$ are the same as in $\widetilde{T}$. Equivalently, $\widetilde{T}^{(i)}$ is the tree obtained by rerooting $T_{y_i}$ at the vertex $r_i$. Each of its leaves is either an original $\widetilde{T}$-leaf or a $p_j$ vertex for some $j$. The crux is that for each $x \in C_i$ it holds that $\Sketch_{\widetilde{G}} (V(\widetilde{T}_x)) = \oplus_{v \in \widetilde{T}^{(i)}_{x}} \beta_v$. That is, the $\widetilde{T}$-subtree $\widetilde{G}$-sketch of $x$ is equal to the sum-of-$\beta$'s in its $\widetilde{T}^{(i)}$-subtree.
Hence, we complete the computation in this step by executing bottom-up XOR-aggregation of the $\beta_v$ values in each of the trees $\widetilde{T}^{(i)}$ in parallel.

Finally, as $\Sketch_{\widetilde{G}} (V\setminus \{y\})$ is just the all-$0$'s string (by \Cref{obs:0-sketch}), it is also known to each $x \in C_1 \cup \cdots \cup C_k$.

\paragraph{Step 4: Local Bor\r{u}vka Simulation In $G \setminus \{x,y\}$.}
This entire step is executed by local computation in which each $x \in C_1 \cup \cdots \cup C_k$ determines whether it is a cut vertex in $G \setminus \{y\}$, or equivalently if $xy$ is a cut pair in $G$. This is done by locally simulating Bor\r{u}vka's algorithms using the sketches of the components of $\widetilde{T} \setminus \{x\}$ (which are known to $x$ by Step 3) in an identical manner to the last step of the (single) cut vertices detection algorithm of Section \ref{sec:single-cut}, replacing $G$ and $T$ there with $G \setminus \{y\}$ and $\widetilde{T}$. We note that the new ancestry labels, extended identifiers and sketches, \emph{computed w.r.t.\ $\widetilde{T}$}, are important for this simulation to follow through exactly as in Section \ref{sec:single-cut}. This completes the proof of Lemma \ref{lem:Ay}.

We conclude this section by describing the scheduling of the algorithms $\{A_y\}_{y \in V}$: 

\begin{lemma}
	The collection of algorithms $\{\cA_y\}_{y \in V}$ can be executed simultaneously within $\widetilde{O}(D)$ rounds, w.h.p.
\end{lemma}
\begin{proof}
	The key observation is that every edge $e$ participates in $O(D)$ algorithms. Specifically, since each algorithm $\cA_y$ exchanges messages only on edges incident to $V(T_y)$, we get that the algorithms using $e=(u,v)$ are exactly $\{\cA_y \mid y \in \pi(s,u,T) \cup \pi(s,v,T) \}$. Therefore, the total number of messages sent through $e = (u,v)$ in the collection of $n$ algorithms $\{A_y\}_{y \in V}$ is at most $\widetilde{O}(1)\cdot (\vert \pi(s,u, T)\vert +\vert \pi(s,v, T)\vert )=\widetilde{O}(D)$. The result follows by employing Theorem \ref{thm:delay} with congestion and dilation bounds of $\widetilde{O}(D)$.
\end{proof}

\subsection{Implementation Details for \Cref{sec:depend}}\label{sec:depend-imp-details}
\APPENDNEWANCESTRYLABELS

\section{Independent Cut Pairs} \label{sec:indep}

We now turn to consider the case where the cut pair $xy$ is independent, i.e., $x,y$ have no ancestor/descendant relations. Throughout this section, for every vertex $x \in V$, let $V_x=V(T_x)\setminus \{x\}$. Recall that we assume that there is no single cut vertex in the graph. Our algorithm is based on the introduced notion of $x$-\emph{connectivity trees}, $\widehat{T}_x$, computed locally at each vertex $x$. Let $\mathcal{C}_x=\{C_1, \ldots, C_k\}$ denote the maximal connected components in the induced graph $G[V_x]$. 
For each $C \in \mathcal{C}_x$, the tree $\widehat{T}_x$ contains a path $\pi_x(s,C)=\pi(s,u_C) \circ (u_C,v_C)$, where $(u_C, v_C)$ is a $G$-edge such that $v_C \in C$, and $x \notin \pi(s,u_C)$. Therefore, $\widehat{T}_x$ encodes the connectivity of $s$ to $V_x$ in the graph $G \setminus \{x\}$. For every $v \in V_x$, let $C_{x,v}$ denote the component containing $v$ in $\mathcal{C}_x$.  When $v=x_h$, we denote $H_x=C_{x,x_h}$ and call it the \emph{heavy component} of $x$. See illustration in \Cref{fig:connectivity-tree}.
We next describe the computation of these $\widehat{T}_x$ trees, and later on show how they guide the identification of independent cut pairs.

Throughout, we assume that the vertices hold all the preprocessing information as in \Cref{sec:depend}, and that the ID of each vertex $v$ contains also its compressed-path $\LCALabel_T(v) = \pi^*(s,v)$.

\begin{figure*}
	\begin{center}
		\includegraphics[height=5.6cm]{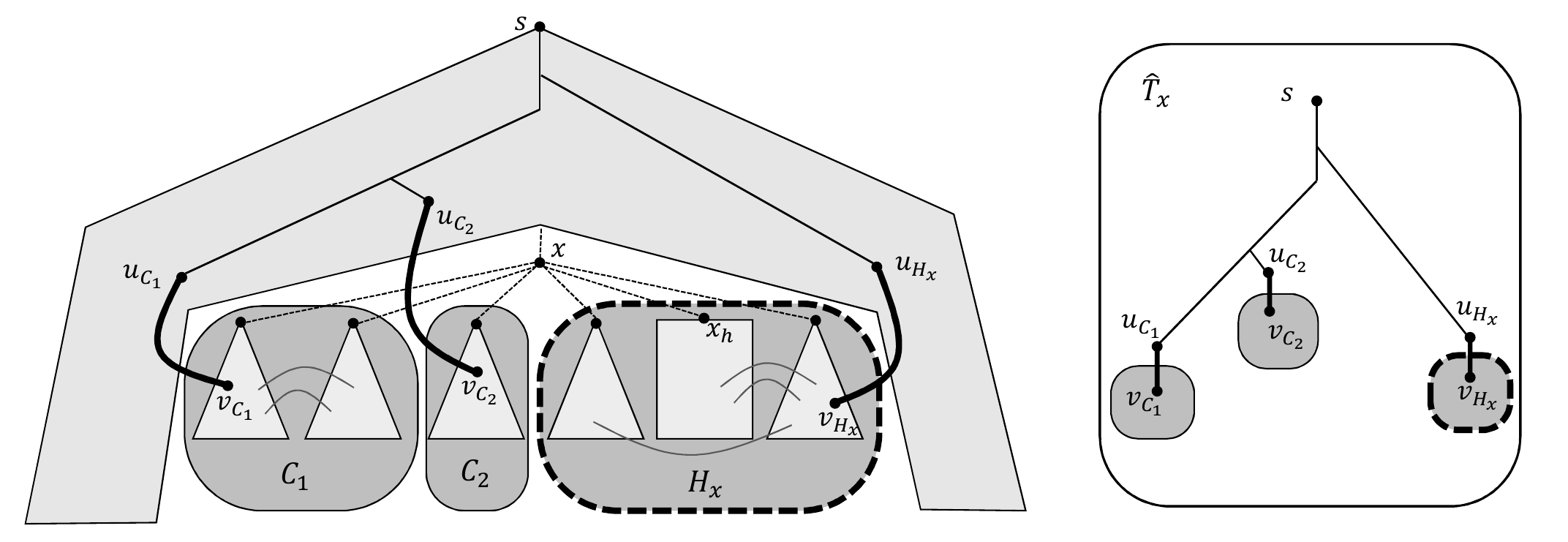}
		\caption{\sf Left: Illustration of $\mathcal{C}_x$, the connected components in $G[V_x]$. Each component is a union of subtrees rooted at children of $x$. The heavy component $H_x$ is the one containing $x_h$. For each component $C \in \mathcal{C}_x$ the corresponding edge $(u_C, v_C)$ appears as a thick solid curve. The tree paths to the $u_C$ vertices, all lying outside $T_x$, are shown by solid lines.
		Right: The corresponding tree $\mathcal{T}_x$, consisting of the paths $\pi_x (s, C) = \pi(s, u_C) \circ (u_C, v_C)$ for each $C \in \mathcal{C}_x$.
		\label{fig:connectivity-tree} 
		}
	\end{center}
\end{figure*}

\subsection{Computing $x$-Connectivity Trees}\label{sec:connectivity-trees}
The computation of connectivity trees has two main steps, both based on bottom-up aggregation of certain graph sketches over the BFS tree $T$. The purpose of the first step is letting every $x \in V$ determine the connected components $\mathcal{C}_x$ in $G[V_x]$. Each such component $C$ is identified by its component-ID, which is defined as the largest vertex ID among all the IDs of $T$-children of $x$ inside $C$. Additionally, each $u \in V$ learns the component-ID of its component $C_{x,u} \in \mathcal{C}_x$, for all of its ancestors $x \in \pi(s,u)$. The second step aggregates a special form of graph sketches, providing $x$ with the path information required to locally compute its connectivity tree $\widehat{T}_x$.

\paragraph{Step 1: Computing Connectivity in $G[V_x]$.} 
For ease of notation, denote by $d_x=\depth_T(x)$ the depth of vertex $x$ in $T$. We say that an edge $e=(u,v)$ has depth $d$ if $\depth_T (\LCA(u,v))=d$. This step is based on aggregating along $T$ the information of $D$ types of graph sketches, one for every depth $1 \leq d \leq D$. The $d^{th}$ sketch type $\Sketch^d_{G}(\cdot)$ will be restricted to sampling only edges of depth \emph{at least} $d$. This is helpful by the following observation: in order to locally simulate the connectivity Bor\r{u}vka algorithm in $G[V_x]$ at every $x$, it is required for $x$ to learn $\Sketch_{G[V_x]}(V(T_{x'}))$ for each $T$-child $x'$ of $x$. Since the edges of $G[V_x]$ can be identified as $G$-edges in $V_x \times (V \setminus \{x\})$ of depth at least $d_x$, $x$ can easily transform $\Sketch^{d_x}_{G}(V(T_{x'}))$ to the required sketch by eliminating the $x$-incident edges from it.
The following lemma summarizes the output of this step, and the detailed implementation appears in its proof.
% All missing proofs in this section appear in Appendix \ref{app:indep}.

\begin{lemma}\label{lem:connecitivity-in-Gxs}
	There is a randomized $\widetilde{O}(D)$-round algorithm that computes connectivity in $G[V_x]$ for every $x \in V$ simultaneously. At the end of its execution, w.h.p.\ all of the following hold:
	\begin{enumerate}
		\item Every $x \in V$ knows the component-ID of $C_{x,x'}$ for each of its children $x'$.
		\item Every $u \in V$ knows the component-ID of $C_{x,u}$ for each of its ancestors $x \in \pi(s,u)$.
		\item Every $v \in V$ knows the information of item 2 for each of its neighbors. That is, the component-ID of $C_{x,u}$ for every $u$ adjacent to $v$ and $x \in \pi(s,u)$.
	\end{enumerate}

\end{lemma}
\def\APPENDCONDEPTHSKETCH{
    \begin{proof}
	% \begin{proof}[Proof of Lemma \ref{lem:connecitivity-in-Gxs}]
	This is implemented as follows. First, within $O(D)$ rounds, we let each vertex learn its tree path $\pi(s,u)$, and exchange this information with all its neighbors. Using this information, $u$ can locally compute for every $1 \leq d \leq D$ the set of its incident edges with depth at least $d$, that is
	$
	E_d(u) = \{(u,v) \in E \mid \depth_T (\LCA(u,v))\geq d\}
	$.
	Then, the source $s$ locally samples a random seed and broadcasts it to all the vertices. Using this seed, each $u$ computes the $D$ sketches $\Sketch^1_{G}(u),\ldots, \Sketch^{D}_{G}(u)$, where the edges in the $d$'th sketch $\Sketch^d_{G}(u)$ are based on using the seed to implement the sampling of the edges only from $E_d(u)$. 
	Next, the algorithm aggregates these $D$ sketch types over subtrees. This can be done in a pipeline manner from the leaf vertices up to the root, in increasing order of the depth $d$ of the sketches. At the end of this computation, each vertex $x$ of depth $d_x$ holds the the sketch $\Sketch^{d_x}_G(V(T_{x'}))$ for each of its $T$-child $x'$. The final sketch $\Sketch_{G[V_x]}(V(T_{x'}))$ is obtained locally at $x$ by canceling-out its own edges to $V(T_{x'})$ from $\Sketch^{d_x}(V(T_{x'}))$, using Lemma \ref{lem:cancel-sketch-property}.
	At this point, each vertex $x$ has all the required sketch information to locally simulate the Bor\r{u}vka algorithm in $G[V_x]$. As a result of this computation, w.h.p.\ $x$ holds the component-ID of $C_{x,x'} \in \mathcal{C}_x$ for each $T$-child $x'$ of $x$, establishing item 1. This information is then propagated down the tree $T$ in a pipeline manner, where each vertex $u$ eventually learns the component-ID of $C_{x,u}$ for each of its ancestors $x \in \pi(s,u)$, establishing item 2. By exchanging the information of item 2 between neighbors, we obtain also item 3.  It is easy to see that this entire computation takes $\widetilde{O}(D)$ rounds.
\end{proof}
}\APPENDCONDEPTHSKETCH

\paragraph{Step 2: Computing $x$-Connectivity Trees $\widehat{T}_x$ via Path-Sketches.} 
Our next goal is to provide each vertex $x$ with the path information $\pi_x(s,C)$, for every component $C \in \mathcal{C}_x$. Recall that $\pi_x(s,C) = \pi(s,u_C) \circ (u_C, v_C)$, where $v_C \in C$ and $x \notin \pi(s,u_C)$. Such a path must exists as we assume $x$ is not a cut vertex. Towards this goal, we define \emph{path-sketches} $\Sketch_{G}^P(\cdot)$. These are defined exactly as the regular sketches, only using new edge identifiers $\EID_T^P(e=(u,v))$ obtained by augmenting $\EID_T (e)$ of Eq.\ \ref{eq:extend-ID} with the tree paths of the endpoints $u,v$. Formally,
%\begin{equation}\label{eq:extend-ID-path}
%	\EID^P_T(e = (u,v))=[\UID(e), \ID(u), \ID(v), \LCALabel_T(u), \LCALabel_T(v),\pi(s,u), \pi(s,v)]~.
%\end{equation}
\begin{align}\label{eq:extend-ID-path}
	\EID^P_T(e = (u,v))= [ \UID(e), \ID(u), \ID(v),
	 \LCALabel_T(u), \LCALabel_T(v),  
	\pi(s,u), \pi(s,v)]~. 
\end{align}
Note that in contrast to $\EID_T (e)$ of \Cref{eq:extend-ID} which contains $\widetilde{O}(1)$ bits, $\EID_T^P (e)$ of \Cref{eq:extend-ID-path} has $\widetilde{O}(D)$ bits. Therefore, a path-sketch also has $\widetilde{O}(D)$ bits. However, we can still aggregate the path-sketches on $T$ so that eventually every $x$ learns $\Sketch^{P}_G(V(T_{x'}))$ for all its children $x'$. Using the information of the previous step, $x$ can now locally compute the path-sketches of each component $C \in \mathcal{C}_x$, and then obtain $\pi_x (s,C)$ by extracting $\EID_T^P(e)$ of an outgoing edge from $C$. The proof of the following lemma explains this process in further detail.

\begin{lemma}\label{lem:path-sketch}
	There is a randomized $\widetilde{O}(D)$-round algorithm letting each vertex $x \in V$ learn its connectivity tree $\widehat{T}_x$, w.h.p.
\end{lemma}
\def\APPENDPATHSKETCH{
    \begin{proof}
	% \begin{proof}[Proof of Lemma \ref{lem:path-sketch}]
		In order to compute the $\EID_T^P$-identifiers, it is enough to let each vertex $u$ learn its tree path $\pi(s,u)$ and exchange this information with all its neighbors, which can be done within $O(D)$ rounds. The source then broadcast an $\widetilde{O}(1)$-bit random seed in $\widetilde{O}(D)$ rounds, allowing each vertex $u$ to compute its path-sketch $\Sketch_{G}^P(u)$. The subtree path-sketches $\Sketch^{P}_G(V(T_u))$ for every vertex $u$ are then computed by XOR-aggregation of $\widetilde{O}(D)$-length vectors on $T$, which can be done in $\widetilde{O}(D)$ rounds via a standard pipeline. At the end of this computation, each vertex $x$ holds the path-sketch $\Sketch^{P}_G(V(T_{x'}))$ for each of its $T$-children $x'$. Using Lemma \ref{lem:cancel-sketch-property}, $x$ can locally cancel-out its own edges to obtain $\Sketch^{P}_{G \setminus \{x\}}(V(T_{x'}))$. Using the information of Lemma \ref{lem:connecitivity-in-Gxs}, $x$ can add-up the subtree sketches of its children to obtain $\Sketch^{P}_{G \setminus \{x\}}(C)$ for all $C \in \mathcal{C}_x$. Specifically, letting $N(x,C)$ denote the set of $T$-children of $x$ inside the component $C$, we have $\Sketch^{P}_{G\setminus \{x\}}(C)=\oplus_{x' \in N(x, C)}\Sketch^{P}_{G\setminus \{x\}}(V(T_{x'}))$. Using $O(\log n)$ basic sketch units of each such component sketch enables $x$ to learn, w.h.p., the $\EID_T^P$-identifier of one outgoing edge per component. As $C$ is a maximal connected component in $G[V_x]$, its outgoing edge must connect it to $V \setminus V(T_x)$. Specifically, for $C \in \mathcal{C}_x$, $x$ learns $\EID_T^P((u_C, v_C))$ where $v_C \in C$ and $u_C \in V \setminus V(T_x)$. As the tree-path $\pi(s,u_C)$ is stored in $\EID_T^P((u_C, v_C))$, the path $\pi_x (s,C)$ is easily deduced. The final tree is given by $\widehat{T}_x=\bigcup_{C \in \mathcal{C}_x} \pi_x(s,C)$.
	\end{proof}
}\APPENDPATHSKETCH

Finally, we would like each vertex $u$ to learn some representation of the path of its component $C_{x,u}$ in the tree $\widehat{T}_x$, for all its ancestors $x \in \pi(s,u)$. To this end, we define compressed representation of these paths: for $x \in V$ and $C \in \mathcal{C}_x$ the compressed path of $\pi_x (s, C)$ is $\pi_x^* (s, C) = \pi^* (s, u_C) \circ (u_C, v_C)$, consisting only of $\widetilde{O}(1)$ bits. This enables us to achieve the following.

\begin{lemma}\label{lem:comp-path-sketch}
	Given that every vertex $x \in V$ holds its connectivity tree $\widehat{T}_x$, there is an $\widetilde{O}(D)$-round algorithm such that in the end of its execution, the following hold:
	\begin{enumerate}
		\item Each $u \in V$ learns the compressed path $\pi_x^*(s, C_{x,u})$ for every $x \in \pi(s,v)$, as well as the entire path $\pi_x(s, C_{x,u})$ in case $x \in \LA(u)$.
		\item Each $v \in V$ learns all the information of item 1 for each neighboring vertex $u$.
	\end{enumerate}
\end{lemma}
\def\APPENDCOMPPATHSKETCH{
% \begin{proof}[Proof of Lemma \ref{lem:comp-path-sketch}]
\begin{proof}
	We let every vertex $x$ send the full path $\pi_x(s,C_{x,x'})$ and its compressed version $\pi_x^*(s,C_{x,x'})$ to each light child $x'$ of $x$, and only the compressed path $\pi^*_x(s,H_x)$ to its heavy child $x_h$. This information is propagated towards the leaf vertices of $T_x$. Since each vertex is required to receive 
	$\widetilde{O}(D)$ bits of information from each of its \emph{light} ancestors, as well as $\widetilde{O}(1)$ bits from each of its heavy ancestors, overall it is required to receive $\widetilde{O}(D)$ bits. This can be done in $\widetilde{O}(D)$ rounds, by standard pipeline techniques. This establishes item 1. For item 2, we simply exchange the  $\widetilde{O}(D)$-bit information of item 1 between neighboring nodes, within additional  $\widetilde{O}(D)$ rounds.
\end{proof}
}\APPENDCOMPPATHSKETCH

\subsection{Component Classification Based on Sensitivity} 
We use the structure of the $x$-connectivity tree $\widehat{T}_x$ to classify its potential independent cut-mates. The first immediate observation is that in case $\hat{T}_x$ is non-empty, any independent cut-mate $y$ of $x$ (if exists) must be in $\widehat{T}_x$. We further examine the \emph{sensitivity} of each component in $\mathcal{C}_x$ to different such potential cut-mates $y$, as described in the following definition.

\begin{definition}[Sensitivity Notions of $\mathcal{C}_x$ Components]\label{def:sensitive-components}
	Fix an independent pair $x,y$. A component $C \in \mathcal{C}_{x}$ is called \emph{$y$-sensitive} if $y \in \pi_x(s,C)$.
	
	We further classify the $y$-sensitive components into two types:
	\begin{itemize}
		\item $C \in \mathcal{C}_{x}$ is \emph{pseudo-$y$-senstive} if $\pi_x(s,C)$ contains a $T$-edge $(y,y')$ such that $x \notin \pi_y(s,C_{y,y'})$. That is, the component $C_{y,y'}$ of $y'$ in $\mathcal{C}_y$ is not $x$-sensitive.
		\item $C \in \mathcal{C}_{x}$ is \emph{fully-$y$-senstive} if it is $y$-sensitive but not pseudo-$y$-sensitive. Equivalently, if either $y = u_C$, or $\pi_x (s,C)$ contains a $T$-edge $(y,y')$ such that $x \in \pi_y(s,C_{y,y'})$.
	\end{itemize}
	We denote by $\mathcal{S}(x,y)$, $\mathcal{PS}(x,y)$ and $\mathcal{FS}(x,y)$ the collections of components in $\mathcal{C}_x$ that are $y$-sensitive, pseudo-$y$-sensitive and fully-$y$-sensitive, respectively\footnote{Notice that these notations are not symmetric in $x,y$, e.g.\ $\mathcal{S}(x,y)$ is different than $\mathcal{S}(y,x)$.}.
\end{definition}

\begin{lemma}\label{lem:non-or-pseudo-sensitive}
	Let $x,y$ be an independent pair, and let $C \in \mathcal{C}_x$. If $C$ is not fully-$y$-sensitive (i.e.\ $C \in \mathcal{C}_x \setminus \mathcal{FS}(x,y)$), then $C$ is connected to $s$ by a path in $G \setminus \{x,y\}$ avoiding all components in $\mathcal{FS}(x,y) \cup \mathcal{FS}(y,x)$.
\end{lemma}
\begin{proof}
	Consider first the case where $C$ is not $y$-sensitive: then $\pi_x (s,C)$ is the required path. The other case is when $C$ is pseudo-$y$-sensitive. Then, let $(y,y') \in \pi_x(s,C)$ be a $T$-edge such that $C_{y,y'} \in \mathcal{C}_y$ is not $x$-sensitive. By the previous case, $C_{y,y'}$ is connected to $s$ by a path in $G \setminus \{x,y\}$ avoiding all components in $\mathcal{FS}(x,y) \cup \mathcal{FS}(y,x)$. Denote this path by $Q$. Now, Observe that the last edge $(u_C, v_C)$ of $\pi_x (s,C)$ is such that $v_C \in C$ and $u_C \in C_{y,y'}$ (as $x$ and $y$ are independent). W.l.o.g., we may assume that the endpoint of $Q$ in $C_{y,y'}$ is $u_C$ (as any two vertices of $C_{y,y'}$ are connected by a path inside $C_{y,y'}$). Thus, $Q \circ (u_C, v_C)$ is the required path for $C$.
\end{proof}

The following lemma shows that each vertex $x$ can, in many cases, distinguish between full sensitivity and pseudo sensitivity of a component $C \in \mathcal{C}_x$ to the potential cut-mates in $\pi_x(s,C)$.

\begin{lemma}\label{lem:learning-strong-sens}
	There is a randomized $\widetilde{O}(D)$-round algorithm such that by the end of the execution, w.h.p.\ each $x \in V$ holds the following information:
	\begin{enumerate}
		\item $\pi^*_y(s, C_{y,y'})$ for every $T$-edge $(y,y') \in \pi_x(s, C)$ and every $C \in \mathcal{C}_x \setminus \{H_x\}$.
		\item $\pi^*_y(s, C_{y,y'})$ for every \emph{light $T$-edge} $(y,y') \in \pi_x(s, H_x)$
	\end{enumerate}
	Note that given $\pi^*_y(s, C_{y,y'})$, $x$ can check whether $x \in \pi_y(s, C_{y,y'})$. Consequently, the following holds: if $C \in \mathcal{S}(x,y)$ such that either $C \neq H_x$ or $(y,y_h) \notin \pi_x (s,C)$, then $x$ knows if $C \in \mathcal{FS}(x,y)$ or $C \in \mathcal{PS}(x,y)$.
\end{lemma}
\def\APPENDLEARNINGSTRONGSEN{
% \begin{proof}[Proof of Lemma \ref{lem:learning-strong-sens}]
\begin{proof}
	We start in a preprocessing step, in which we let each vertex $u$ learn the compressed-path $\pi^*_y(s,C_{y,y'})$ for every edge $(y,y') \in \pi(s,u)$, and further exchange this information with all its neighbors. This can be done by downcasting the information on $T$ within $\widetilde{O}(D)$ rounds. At this point, for every $x \in V$ and $C \in \mathcal{C}_x$, the vertex $v_C \in C$ (which is the last vertex of $\pi_x(s,C)$) holds the compressed path $\pi^*_y(s, C_{y,y'})$ for every $T$-edge $(y,y') \in \pi_x(s, C)$. Our goal is to pass (most) of this preprocessing information up to $x$. To this end, we design for each $x \in V$ two specialized algorithms, $\mathcal{L}_x$ and $\mathcal{H}_x$, letting $x$ learn the information specified in item 1 and item 2 of the lemma, respectively. They have the following properties:
	\begin{itemize}
		\item $\mathcal{L}_x$ has dilation and congestion $\widetilde{O}(D)$, sending messages only on edges incident to $\LD(x)$.
		\item $\mathcal{H}_x$ has dilation $\widetilde{O}(D)$, congestion $\widetilde{O}(1)$, sending messages only on edges incident to $V(T_x)$.
	\end{itemize}
	These enable us to execute all algorithms $\{\mathcal{L}_x, \mathcal{H}_x\}_{x \in V}$ in parallel within $\widetilde{O}(D)$ rounds using Theorem \ref{thm:delay}. Indeed, the dilation bound is obvious. As for congestion, consider some edge $e$. There are at most $O(\log n)$ vertices $x$ such that either of its endpoint belongs to $\LD(x)$, hence $e$ needs to pass at most $O(\log n) \cdot \widetilde{O}(D) = \widetilde{O}(D)$ messages of $\mathcal{L}_x$ algorithms. Also, there are at most $O(D)$ vertices $x$ such that either of its endpoints belongs to $V(T_x)$, hence $e$ needs to pass at most $O(D) \cdot \widetilde{O}(1) = \widetilde{O}(D)$ messages of $\mathcal{H}_x$ algorithms.
	
	It remains to describe the algorithms $\mathcal{L}_x$ and $\mathcal{H}_x$. For $\mathcal{L}_x$, we let each vertex $v_C$ of a non-heavy component $C \in \mathcal{C}_x \setminus \{H_x\}$ propagate the entire $\widetilde{O}(D)$-bit preprocessing information upwards to $x$ on $\pi(x, v_C)$. As the collection of paths $\{\pi(s,v_C)\}_{C \in \mathcal{C}_x \setminus \{H_x\}}$ is edge-disjoint, this can be done in parallel on all these paths, yielding the required dilation and congestion for $\mathcal{L}_x$. Finally, in $\mathcal{H}_x$, we let $v_{H_x}$ propagate the $\widetilde{O}(1)$-bits of relevant information for item 2 in the lemma up to $x$ along $\pi(x,v_{H_x})$, which clearly achieves the required dilation and congestion for $\mathcal{H}_x$.
\end{proof}
}\APPENDLEARNINGSTRONGSEN

From now on, we assume that all the (randomized) algorithms specified in lemmas \ref{lem:connecitivity-in-Gxs}, \ref{lem:path-sketch}, \ref{lem:comp-path-sketch} and \ref{lem:learning-strong-sens} were executed, and gave correct outputs (which happens w.h.p.).

\subsection{$xy$-Connectivity Algorithms Under a Promise}\label{sec:promise}
We next discuss the procedure for determining the $xy$-connectivity (i.e.\ the connectivity of $G \setminus \{x,y\}$) for collections of independent pairs $xy$ which satisfy a given \emph{promise}. For an overview, see \Cref{sec:approach}.
The following definition plays a key role in this procedure.

\begin{definition}
	Let $x,y$ be an independent pair. We define $\LDS(x,y)$ as the set of \textbf{l}ight $x$-\textbf{d}escendants in $y$-\textbf{s}ensitive components. That is,
	$
		\LDS(x,y) =\{ v \in \LD(x) \mid y \in \pi_x(s, C_{x,v})\}.
	$
\end{definition}

%Recall that a vertex $v$ has $O(\log n)$ light ancestors, i.e.\ can belong to $\LD(x)$ for $O(\log n)$ vertices $x$. For each such $x$, we have $\vert \pi_x(s, C_{x,v})\vert  = O(D)$. We immediately get:
\begin{observation}\label{obs:bound-lds}
	Every vertex $v$ belongs to a total of $O(D\log n)$ sets $\LDS(x,y)$ for $x,y \in V$.
\end{observation}
\begin{proof}
	$v$ belongs to $\LD(x)$ only for its $O(\log n)$ light ancestors $x$, and $\vert \pi_x (s, C_{x,v})\vert  = O(D)$.
\end{proof}

The following theorem states the properties of a single $xy$-connectivity algorithm.

\begin{theorem}[$xy$-Connectivity Given an $x$-$y$ Path]\label{lem:connectivity-with-promise}
	Fix an independent pair $x,y$, and assume that there is an $x$-$y$ path $\Pi_{x,y} \subseteq G$ of length $O(D)$ (known in a distributed manner). Then, there is a randomized $xy$-connectivity algorithm $\mathcal{A}^P_{x,y}$ of $\widetilde{O}(D)$ rounds and $\widetilde{O}(1)$ congestion, sending messages only along (i) the edges of $\Pi_{x,y}$ and (ii) edges incident to $\LDS(x,y) \cup \LDS(y,x)$. At the end of the execution, w.h.p.\ both $x$ and $y$ know whether $G \setminus \{x,y\}$ is connected or not.
\end{theorem}

The major part of this section is devoted to proving \Cref{lem:connectivity-with-promise}. Before doing so, we show that for a set of independent pairs $Q \subseteq V \times V$, all algorithms $A_{x,y}^P$ for $xy \in Q$ can be scheduled simultaneously when provided a path collection $\mathcal{P}_Q = \{\Pi_{x,y} \mid xy \in Q\}$ satisfying the following:

\begin{quote}[\textbf{Promise:}]
	\emph{$\mathcal{P}_Q$-paths have length $O(D)$, and each $G$-edge appears on $\widetilde{O}(D)$ \mbox{$\mathcal{P}_Q$-paths}.}
\end{quote}

\begin{corollary}\label{cor:allpairsconnectivity-with-promise}[All Pairs $xy$-Connectivity Under a Promise]
	Let $Q \subseteq V \times V$ be a collection of independent pairs, and $\mathcal{P}_Q=\{\Pi_{x,y} \mid xy\in Q\}$ be a collection of $x$-$y$ paths (each known is a distributed manner) satisfying the promise.
	Then, all algorithms $\mathcal{A}^P_{x,y}$ for $xy \in Q$ (where $\mathcal{A}^P_{x,y}$ uses the path $\Pi_{x,y} \in \mathcal{P}_Q$) can be executed simultaneously within $\widetilde{O}(D)$ rounds, w.h.p.
\end{corollary}
\def\APPENDCORALLPAIRSCON{
\begin{proof}
%\begin{proof}[Proof of Cor. \ref{cor:allpairsconnectivity-with-promise}]
	We show that the total congestion of these algorithms is $\widetilde{O}(D)$, which immediately yields the result by \Cref{thm:delay}. Consider an edge $e$. By \Cref{obs:bound-lds}, there are at most $\widetilde{O}(D)$ pairs $xy \in Q$ such that $e$ is incident to sets $\LDS(x,y) \cup LDS(y,x)$. Also, by the promise, $e$ is present in $\Pi_{x,y}$ only for $\widetilde{O}(D)$ pairs $xy \in Q$. Thus, $e$ is used by at most $\widetilde{O}(D)$ algorithms, each of sending $\widetilde{O}(1)$ messages through it, so the congestion is $\widetilde{O}(D)$.
\end{proof}
}\APPENDCORALLPAIRSCON

\paragraph{Description of the Connectivity Algorithm $\mathcal{A}^P_{x,y}$.}
The algorithm is based on simulating the Bor\r{u}vka algorithm using the sketch information of connected subsets in $G \setminus \{x,y\}$, held jointly by $x$ and $y$. Throughout, we refer to the given $x$-$y$ path $\Pi_{x,y}$ as \emph{the $xy$-channel}

The input for phase $i\geq 1$ of Bor\r{u}vka is the following.
There is a partitioning $\mathcal{P}_{i-1}=\{P_{i-1,1},\ldots, P_{i-1,k_{i-1}}\}$ of the vertices in $V \setminus \{x,y\}$ into connected subsets, called \emph{parts} (to avoid confusion with the `components' in $\mathcal{C}_x$ and $\mathcal{C}_y$).
We mark a special vertex in each $P_{i,j} \in \mathcal{P}_i$, called the \emph{leader} of the part. 
The leaders are either some chosen $T$-children of $x$ or $y$ in these parts, or (in some cases) the source $s$. The part-ID is the ID of its leader. 
The part containing $s$ is called the \emph{$s$-part}.
The part containing $x_h$ (resp., $y_h$) is called $x$-\emph{heavy} (resp., $y$-\emph{heavy})\footnote{A part can be both $x$-heavy and $y$-heavy.}. The parts that are free of $s,x_h,y_h$ are called \emph{light}.
In the initial partitioning $\mathcal{P}_0$, we make sure that all light parts are contained in $\LDS(x,y) \cup \LDS(y,x)$ (as will be described later). As the parts only get merged throughout the Bor\r{u}vka execution, this remains true also for light parts in $\mathcal{P}_{i-1}$.
A part $P$ is said to be \emph{growable} if there is an outgoing $G$-edge connecting $P$ to $V \setminus (P \cup \{x,y\})$.
The Bor\r{u}vka algorithm has $K=O(\log n)$ forest growing phases in $G \setminus \{x,y\}$, where each phase reduces the number of growable parts by a constant factor, in expectation. We maintain the following invariants for the beginning of each phase $1\leq i \leq K$, for both $z \in \{x,y\}$:

\begin{enumerate}[label={(I\arabic*)}]
	\item $z$ knows $\Sketch_{G \setminus \{x,y\}}(P)$ and the part-ID of the \emph{$s$-part} $P \in  \mathcal{P}_{i-1}$.
	\item $z$ knows $\Sketch_{G \setminus \{x,y\}}(P)$ for every \emph{light part} $P \in  \mathcal{P}_{i-1}$ whose leader is in $T_{z}$. 
	\item $z$ knows $\Sketch_{G \setminus \{x,y\}}(P)$ as well as the part-IDs of all \emph{heavy parts} $P \in \mathcal{P}_{i-1}$.
	\item $z$ knows, for each $T$-child $z'$ of $z$, the part-ID of the part containing $z'$ in $\mathcal{P}_{i-1}$.
\end{enumerate}

\paragraph{Initialization.}
We describe the initialization executed prior to the first Bor\r{u}vka phase. The initial parts are defined as the fully-$y$-sensitive components in $\mathcal{C}_x$, the fully $x$-sensitive components of $\mathcal{C}_y$, and another part containing all the rest of the vertices in $G \setminus \{x,y\}$ (and particularly the source $s$). That is, the initial partitioning is
$\mathcal{P}_0 = \mathcal{FS}(x,y) \cup \mathcal{FS}(y,x) \cup \{U(x,y)\}$,
where the $s$-part $U(x,y)$ is
$$
U(x,y) = V \setminus \left[ \{x,y\} \cup \left(\bigcup_{C \in \mathcal{FS}(y,x) \cup \mathcal{FS}(y,x)} C \right) \right]
.$$
Parts from $\mathcal{FS}(x,y)$ or $\mathcal{FS}(y,x)$ are clearly connected, and $U(x,y)$ is connected by \Cref{lem:non-or-pseudo-sensitive}. Note that light parts can come only from $\mathcal{FS}(x,y) \cap \LD(x) \subseteq \LDS(x,y)$ or from  $\mathcal{FS}(y,x) \cap \LD(y) \subseteq \LDS(y,x)$.
The leader of each part $P \in \mathcal{FS}(x,y)$ is chosen as the vertex of largest ID among all $T$-children of $x$ in $P$, and similarly for parts in $\mathcal{FS}(y,x)$. The source $s$ is the leader of $U(x,y)$.

The following claim builds on \Cref{lem:learning-strong-sens} and using the $xy$-channel.
All missing proofs in this section are deferred to \Cref{sec:promise-imp-details} on implementation details.

\begin{claim}\label{cl:classifying-sensitivity}
	By passing $\widetilde{O}(1)$ messages on the $xy$-channel, $x$ (resp., $y$) can classify all components in $\mathcal{C}_x$ (resp., $\mathcal{C}_y)$ as non-, pseudo- or fully-$y$-sensitive (resp., $x$-sensitive).
\end{claim}
\def\APPENDCLASSIFYINGSENSITIVITY{
\begin{proof}[Proof of \Cref{cl:classifying-sensitivity}]
%\begin{proof}
	W.l.og., we show this only for $x$. Recall that $x$ can classify each component in $\mathcal{C}_x$ as $y$-sensitive or not using $\widehat{T}_x$. Furthermore, by \Cref{lem:learning-strong-sens}, each $y$-sensitive component $C \neq H_x$ can be further classified as pseudo- or fully-$y$-sensitive. It remains to classify $H_x$. By \Cref{lem:learning-strong-sens}, a problem in doing so occurs only if $(y,y_h) \in \pi_x (s, H_x)$. In this case, to classify $H_x$ as pseudo- or fully-$y$-sensitive, it is enough for $x$ to learn $\pi_y^* (s, C_{y,y_h}) = \pi_y^* (s, H_y)$. This $\widetilde{O}(1)$-bit information can be passed to $x$ from $y$ over the $xy$-channel. 
\end{proof}
}%\APPENDCLASSIFYINGSENSITIVITY

The sensitivity classification of the last claim easily enables us to ``almost" fulfill invariants (I1-4) w.r.t.\ $\mathcal{P}_0$, only with $G$-sketches instead of $(G \setminus \{x,y\})$-sketches. Formally, we define the \emph{modified} conditions (I*1-4) exactly the same as (I1-4), only replacing $\Sketch_{G\setminus\{x,y\}} (\cdot)$ with $\Sketch_{G} (\cdot)$ everywhere, and we have:

\begin{claim}\label{cl:modified-invariants}
	Conditions (I*1-4) can be satisfied for both $z \in \{x,y\}$, by passing only $\widetilde{O}(1)$ messages on the $xy$-channel.
\end{claim}
\def\APPENDMODIFIEDINVARIANTS{
\begin{proof}[Proof of \Cref{cl:modified-invariants}]
%\begin{proof}
	W.l.o.g., we show this for $z = x$.
	\begin{enumerate}[label={(I*\arabic*)}]
		\item The sketch of the $s$-part in $\mathcal{P}_0$ is
		\begin{align*}
			\Sketch_G (U(x,y)) = 
			&\Big[ \Sketch_G(V) \oplus \Sketch_G (x) \oplus \Big( \bigoplus_{C \in \mathcal{FS}(x,y)} \Sketch_G (C) \Big) \Big] \\
			\oplus 
			&\Big[ \Sketch_G (y) \oplus \Big( \bigoplus_{C \in \mathcal{FS}(y,x)} \Sketch_G (C) \Big)  \Big] .
		\end{align*}
		The first and second square-bracketed terms can be computed locally by $x$ and $y$, respectively, by \Cref{cl:classifying-sensitivity}. By exchanging these terms over the $xy$-channel, and adding them up, both $x,y$ learn $\Sketch_G (U(x,y))$. The part-ID is the ID of $s$.
		
		\item A light part in $\mathcal{P}_0$ whose leader is in $T_x$ is simply a component $C \in \mathcal{FS}(x,y)$, and $x$ holds the $G$-sketches of such components.
		\item We first let $x,y$ exchange $\Sketch_{G} (H_x)$ and $\Sketch_{G} (H_y)$, and the largest ID of their $T$-children in $H_x$ and $H_y$, over the $xy$-channel. Now, a heavy part in $\mathcal{P}_0$ can be either $U(x,y)$, $H_x$ or $H_y$, so its $G$-sketch and its part-ID is known to both $x,y$. 
		\item Let $x'$ be a $T$-child of $x$. If $C_{x,x'} \in \mathcal{FS}(x,y)$ (which $x$ knows by \Cref{cl:classifying-sensitivity}), then the part-ID for $x'$ is the largest vertex ID of a $T$-child of $x$ in $C_{x,x'}$. Else, $C_{x,x'} \subseteq U(x,y)$, so the part-ID for $x'$ is the ID of $s$.
	\end{enumerate}
\end{proof}
}%\APPENDMODIFIEDINVARIANTS

Finally, the following technical claim enables us to end the initialization step of $\mathcal{A}_{x,y}^P$ by transforming the $G$-sketches to $(G \setminus \{x,y\})$-sketches, thus obtaining (I1-4) w.r.t.\ $\mathcal{P}_0$.

\begin{claim}\label{cl:init-sketch-cancel}
	Assume that the modified invariants (I*1-4) hold w.r.t.\ $\mathcal{P}_0$. Then, within $\widetilde{O}(D)$ rounds with congestion $\widetilde{O}(1)$, where messages are sent only on the $xy$-channel and on edges incident to $\LDS(x,y) \cup \LDS(y,x)$, invariants (I1-4) can be fulfilled w.r.t.\ $\mathcal{P}_0$. 
\end{claim}
\def\APPENDINITSKETCHCANCEL{
\begin{proof}[Proof of \Cref{cl:init-sketch-cancel}]
%\begin{proof}
	We show (I1-4) for $x$, and the proof for $y$ is symmetric.
	We start by letting all $y$-adjacent vertices in $\LDS(x,y)$ send messages to $y$ informing it that they belong to $\LDS(x,y)$ (which they know by \Cref{lem:comp-path-sketch}). This suffices for $y$ to locally compute the following information (where $E(y,S)$ denotes the set of edges connecting $y$ to a vertex in $S \subseteq V$):
		\begin{enumerate}
		\item $\CanSketch(E(y, V(T_{x'})))$ for every $T$-child $x'$ of $x$ such that $x' \in \LDS(x,y)$.
		\item $\CanSketch(E(y, V(T_{x_h})))$.
		\item $\CanSketch( E(y, \LD(x) \setminus \LDS(x,y)) )$.
		\item $\CanSketch( E(y, U(x,y) \setminus V(T_x)) )$.
	\end{enumerate}
	Next, we pass the above information to $x$, as follows. Items 2-4 consist only of $\widetilde{O}(1)$ bits, so they can be sent through the $xy$-channel. To handle item 1, we let $y$ send $\CanSketch(E(y, V(T_{x'})))$ to an arbitrary neighbor $v_{x'} \in T_{x'}$, for each $x'$ which is a $T$-child of $x$ satisfying $x' \in \LDS(x,y)$ (and $E(y, V(T_{x'})) \neq \emptyset$). These cancel-sketches are propagated upwards from the $v_{x'}$ vertices to $x$ in parallel, on the disjoint tree paths $\pi(x, v_{x'})$, which contain only edges incident to $\LDS(x,y)$. By the end of this process, $x$ has all the information of item 1.
	
	We are now ready to show (I1-3) (note that (I4) is trivial as it is the same as (I*4)).
	
	Consider first a light part $P \in \mathcal{FS}(x,y)$. Then $x$ can XOR its sketch $\Sketch_{G}(P)$ with $\CanSketch(E(y, V(T_{x'})))$ for every $T$-child $x'$ of $x$ such that $x' \in P$, using (I*4) and the information of item 1, thus eliminating all $y$-incident edges of $P$. For the same $x'$ vertices, it can further XOR this sketch with $\CanSketch(E(x, V(T_{x'})))$, which is easily computed locally. This eliminates also the $x$-incident edges of $P$, yielding $\Sketch_{G\setminus \{x,y\}} (P)$. This shows (I2).
	
	Next, consider the case where $H_x \in \mathcal{FS}(x,y)$, namely $H_x$ is the heavy part. Then $H_x$ is the union of the vertices in $T_{x_h}$ and in some subtrees $T_{x'}$ of $T$-children $x'$ of $x$ such that $x' \in \LDS(x,y)$. Therefore, $x$ can compute $\Sketch_{G\setminus \{x,y\}} (H_x)$ in the same manner done for light parts, but also using the information of item 2. If $H_y \in \mathcal{FS}(y,x)$, then a symmetric procedure can be applied by $y$, and then $\Sketch_{G\setminus \{x,y\}} (H_y)$ is sent to $x$ over the $xy$-channel. This shows (I3) (except for cases where the $s$-part is also a heavy part, but these are treated by (I1) next).
	
	Finally, consider the $s$-part $U(x,y)$. We show how $x$ can compute the cancellation-sketch of all edges $E(y, U(x,y))$. First, $x$ adds-up the cancellation-sketches of edges in $E(y, \LD(x) \setminus \LDS(x,y)) )$ and $E(y, U(x,y) \setminus V(T_x))$ from items 3 and 4. Next, in case $x_h \in U(x,y)$ (which $x$ knows by (I*4)), $x$ also adds the cancellation-sketch of $E(y, V(T_{x_h}))$ from item 3. At this point, it remains to add the cancellation-sketches of edges going from $y$ to $V(T_{x'})$ for every light $T$-child $x'$ of $x$ such that $C_{x,x'}$ is pseudo-$y$-sensitive (and therefore, $x' \in \LDS(x,y)$). Note that $x$ can determine all such $x'$ children using \Cref{cl:classifying-sensitivity} and (I*4), and add up the corresponding cancellation-sketches from item 1. This yields $\CanSketch(E(y, U(x,y)))$.
	A symmetric procedure allows $y$ to compute $\CanSketch(E(x, U(x,y)))$, and then send it to $x$ through the $xy$-channel. By XORing these last two cancellation sketches with $\Sketch_{G} (U(x,y))$, $x$ can obtain $\Sketch_{G\setminus \{x,y\}} (U(x,y))$. This shows (I1).
\end{proof}
}%\APPENDINITSKETCHCANCEL

\paragraph{Simulation of the $i^{th}$ Bor\r{u}vka Phase.} We now describe the execution of phase $i\geq 1$, assuming that at the beginning of this phase conditions (I1-4) hold w.r.t $\mathcal{P}_{i-1}$. The output of the phase will be the partitioning $\mathcal{P}_i$, for which we will show conditions (I1-4) hold as well.
Our goal is to let $x,y$ simulate a Bor\r{u}vka phase in which parts of $\mathcal{P}_{i-1}$ are merged along their outgoing edges. The main objective of this phase is to reduce the number of \emph{growable} parts by a constant factor, in expectation.
For efficiency of computation, we restrict the merges to have star shapes using random coins (see e.g., \cite{GhaffariH16}). Such star merges are obtained by letting each part of $\mathcal{P}_{i-1}$ toss a random coin, and allowing only merges centered on head-parts, each accepting incoming suggested merge-edges from tail-parts. The leader of this head-part becomes the leader of the merged part.

We now discuss the distributed implementation of the merges by the algorithm $\mathcal{A}_{x,y}^P$. Throughout, we use the following auxiliary claim which allows the vertices in every light part to exchange $\widetilde{O}(1)$ bits, in parallel. 
\begin{claim}\label{cl:light-communication}
    Assume that every vertex $v$ belonging to a light part in $\mathcal{P}_{i-1}$ holds a $\widetilde{O}(1)$-bit value $val(v)$. Within $\widetilde{O}(D)$ rounds with $\widetilde{O}(1)$ congestion, where communication is restricted to edges incident to $\LDS(x,y) \cup \LDS(y,x)$, the vertices of each light part $P$ can learn any aggregate function of the values in their part $\{val(v) \mid v\in P\}$.
\end{claim}
\def\APPENDLIGHTCOMM{
	\begin{proof}[Proof of Claim \ref{cl:light-communication}]
	%\begin{proof}
		First, all vertices belonging to light parts can be informed of their part-ID, by propagating the information of invariant (I4) from $x,y$ to the subtrees of their children in light parts. By exchanging this information along neighbors, each vertex in a light part learns all its neighbors with the same part-ID.
		We will give a procedure for computing the aggregated value of any specific light part $P \in \mathcal{P}_{i-1}$, using only edges from $P \times (P \cup \{x,y\})$. Note that the edges used for different light parts are disjoint, and consist only of edges incident to $\LDS(x,y) \cup \LDS(y,x)$. Hence, the procedures for all light parts can be executed in parallel.
		
		It remains to describe the aforementioned procedure for some light part $P \in \mathcal{P}_{i-1}$. Let $P_x = P \cap V_x$ and $P_y = P \cap V_y$. Note that $P_x$ is the union of $T_{x'}$ subtrees, for the $T$-children $x'$ of $x$ in $P$. We start by aggregating the $val(v)$ values, in parallel, in each such $T_{x'}$ subtree. Each such $x'$ then sends the aggregated value of $T_{x'}$ to $x$. This allows $x$ to compute the aggregated value of $P_x$, which we denote by $val(P_x)$. This value is then broadcasted down the same subtrees, so it is learned by all vertices in $P_x$. A symmetric procedure allows the vertices of $P_y$ to learn $val(P_y)$. If any of $P_x$ or $P_y$ are empty, then we are done. Otherwise, $val(P_x)$ and $val(P_y)$ are exchanged along the edges in $P_x \times P_y$ (at least one such edge must exist, as the part $P$ is connected). This information is propagated upwards to $x$ and $y$ by the endpoint of these edges (ignoring repeated messages). Then, both $x,y$ can compute the final aggregated value $val(P)$, and broadcast it in parallel on each of their subtrees contained in $P$, concluding the procedure.
	\end{proof}
}%\APPENDLIGHTCOMM

We divide the responsibility over the parts of $\mathcal{P}_{i-1}$ between $x$ and $y$ by letting each $z \in \{x,y\}$ be responsible for the parts whose leader is inside $T_z$. W.l.o.g., we make $x$ also responsible for the $s$-part in $\mathcal{P}_{i-1}$ in case its leader is $s$. We start by letting each $z \in \{x,y\}$ toss (locally) a fresh random coin for each of the parts under its responsibility. Denote by $\mathcal{P}_{i-1,z}^H$ and $\mathcal{P}_{i-1,z}^T$ the head and tail parts, respectively, for which $z$ is responsible. Next, $z$ locally computes an outgoing edge $(u_P, v_P)$ for each of its tail-parts $P \in \mathcal{P}_{i-1,z}^T$. Whenever $P$ is growable, such an edge is detected using $\Sketch_{G \setminus \{x,y\}} (P)$ with constant probability. The parts of $\mathcal{P}_i$ are formed by merging every head-part $P^*$ with all the tail-parts in $P$ whose outgoing edges point at $P^*$. The leader of the merged part is defined as the leader of the head-part $P^*$.
The following claim enables learning the part-IDs and coin-tosses of the parts to which the outgoing edges point.

\begin{claim}\label{cl:detect-legal-IDs}
	Within $\widetilde{O}(D)$ rounds with $\widetilde{O}(1)$ congestion, where communication is restricted to edges incident to $\LDS(x,y) \cup \LDS(y,x)$ and the $xy$-channel, each $z \in \{x,y\}$ can determine for all its tail-parts $P \in \mathcal{P}_{i-1,z}^T$  with an outgoing edge $e_P=(u_P,v_P)$, the following information: (i) the part-ID of the second endpoint $v_P \notin P$ and, (ii) the coin-toss of the part of $v_P$. 
\end{claim} 
\def\APPENDCLLEGAL{
	\begin{proof}[Proof of Claim \ref{cl:detect-legal-IDs}]
%	\begin{proof}
		We show this for $z = x$, and the proof for $z = y$ is similar.
		Consider first the (at most three) non-light tail-parts $P \in \mathcal{P}_{i-1,x}^T$. By using the $\LCALabel_T$-labels of the endpoints of $e_P$ (found in the $\EID_T (e_P)$), $x$ can determine if they lie in $V(T_x)$, in $V(T_y)$, or in $V \setminus (V(T_x) \cup V(T_y))$. The part-IDs of endpoints in $V(T_x)$ can be determined by $x$ using (I4). The part-ID of endpoints in $V \setminus V(T_x) \cup V(T_y)$ is the part-ID of the $s$-part, known to $x$ by (I1). In order to learn the part-IDs of the remaining endpoints in $V(T_y)$, $x$ sends all (at most six) such endpoints to $y$ through the $xy$-channel, and $y$ can determine their part-IDs using (I4) and report them back to $x$. At this point, $x$ has the information of (i) for its non-light tail-parts, i.e.\ it knows to which parts their outgoing edges point. If the coin-toss of such a pointed part is not known to $x$, then it can query $y$ for it using the $xy$-channel. This establishes (ii) for non-light $P \in \mathcal{P}_{i-1,x}^T$.
		
		It remains to handle the light tail-parts in $\mathcal{P}_{i-1,x}^T$. First, we let each vertex in $\LDS(x,y) \cup \LDS(y,x)$ learn its part-ID. This is done by using (I4) and broadcasting the information from $x,y$ along the subtrees of $T_x$ of $T_y$ whose union consists of $\LDS(x,y) \cup \LDS(y,x)$. Next, each $z \in \{x,y\}$ sends to the leaders of \emph{light} parts lying in $T_z$ their part-ID, coin-toss, and also the outgoing edge in case the latter is tail. This information is then broadcasted on all light parts in parallel using \Cref{cl:light-communication}. Within additional $\widetilde{O}(1)$ rounds, the $u_P$ vertices of light tail-parts inform their $v_P$ neighbors of this information. The $v_P$ vertices then reply by one of the three following options, where $P'$ is the part containing $v_P$.
		\begin{enumerate}
			\item The part-ID and coin-toss of $P'$. This happens if $P'$ is a light part.
			\item Only the part-ID of $P'$. This happens if $v_P \in \LDS(x,y) \cup \LDS(y,x)$, but $P'$ is non-light.
			\item A null message. This happens in all remaining cases.
		\end{enumerate}
		The replies, along with the identity of the $v_P$ vertices, are then broadcasted on the light parts using \Cref{cl:light-communication}, and sent to $x$ from the leaders of each tail-part $P \in \mathcal{P}_{i-1,x}^T$. If $v_P$ replied with option 1, then clearly $x$ learns both (i) and (ii) for $P$. As $x,y$ can share the coin-tosses of all (at most three) non-light parts on the $xy$-channel, $x$ can learn (i) and (ii) also if $v$ replied with option 2. Finally, assume $v_P$ responded with option 3. In case $v_P \in V(T_{x_h})$ (which $x$ can detect using $\LCALabel_T (v_P)$), then $v_P$ belongs to the $x$-heavy part. If $v_P \in V(T_{y_h})$, then $v_P$ belongs to the $y$-heavy part. Otherwise, $v_P$ must belong to $U(x,y)$, and thus to the $s$-part in $\mathcal{P}_{i-1}$, as it can either be outside $V(T_x)$ and $V(T_y)$, or in a non-sensitive component in $\mathcal{C}_x \cup \mathcal{C}_y$. As the part-IDs of the non-light parts are known to $x$ by (I1) and (I3), and their coin-tosses are also known (as in option 2), this yields (i) and (ii) also when $v_P$ responds with option 3.
	\end{proof}
}%\APPENDCLLEGAL

The main issue in implementing the merges creating $\mathcal{P}_{i}$ is letting each $z \in \{x,y\}$ to learn the updated sketch information of the new parts formed by the merges, each identified with its center which is the corresponding head-part in $\mathcal{P}_{i-1,z}^H$. This is shown in the following claim. For an illustration of the proof, see \Cref{fig:promise}.
\begin{claim}\label{cl:merges}
	Within $\widetilde{O}(D)$ rounds with $\widetilde{O}(1)$ congestion, where communication is restricted to edges incident to $\LDS(x,y) \cup \LDS(y,x)$ and the $xy$-channel, each $z \in \{x,y\}$ can learn the $(G \setminus \{x,y\})$-sketch of each part of $\mathcal{P}_i$ whose center is in $\mathcal{P}_{i-1,z}^H$.
\end{claim}
\begin{proof}
	We show this for $z = x$, and the proof for $y$ is symmetric. We consider two types of star centers: non-light and light.
	
	\textbf{Non-Light Star Centers.} These correspond to the non-light parts in $P_{i-1,x}^H$. There are at most three such parts (at most two heavy parts, and the $s$-part). For each such non-light $P^* \in P_{i-1,x}^H$, the sketch of its corresponding star in $\mathcal{P}_i$ is
%	$$
%	\Sketch_{G \setminus \{x,y\}} (P^*)
%	\oplus
%	\Big[ \bigoplus_{\substack{P \in \mathcal{P}_{i-1,x}^T \\ v_P \in P^*}} \Sketch_{G \setminus \{x,y\}} (P) \Big]
%	\oplus
%	\Big[ \bigoplus_{\substack{P \in \mathcal{P}_{i-1,y}^T \\ v_P \in P^*}} \Sketch_{G \setminus \{x,y\}} (P) \Big]
%	$$
	\begin{gather*}
	\Sketch_{G \setminus \{x,y\}} (P^*) 
	\oplus
	\Big[ \bigoplus_{\substack{P \in \mathcal{P}_{i-1,x}^T \\ v_P \in P^*}} \Sketch_{G \setminus \{x,y\}} (P) \Big] 
	\oplus
	\Big[ \bigoplus_{\substack{P \in \mathcal{P}_{i-1,y}^T \\ v_P \in P^*}} \Sketch_{G \setminus \{x,y\}} (P) \Big] .
	\end{gather*}
	The first two terms can be computed locally by $x$. The third one, i.e.\ the sketch of all $y$-tail-parts pointing to $P^*$, can be computed by $y$ and then sent to $x$ through the $xy$-channel. Thus, $x$ can compute the sum above and obtain the sketch of the new part for $P^*$.
	
	\textbf{Light Star Centers.} It remains to merge the stars centered at \emph{light} parts in $\mathcal{P}_{i-1,x}^H$. We start by letting $y$ send $x$ the sketches of the (at most three) non-light tail-parts in $\mathcal{P}_{i-1,y}^T$, together with the part-IDs of their star centers, through the $xy$-channel. At this point, in order to compute the sketch of a star centered at a light part $P^* \in \mathcal{P}_{i-1,x}^H$, it is enough for $x$ to learn the sum-of-sketches of all \emph{light} tail-parts in $\mathcal{P}_{i-1}$ whose outgoing edge point at $P^*$, which we denote by $val(P^*)$. Namely, $x$ is required to learn
	$$
	val(P^*) = \bigoplus_{\substack{P \in \mathcal{P}_{i-1,x}^T \cup \mathcal{P}_{i-1,y}^T \\ v_P \in P^* \text{ and $P$ is light}}} \Sketch_{G \setminus \{x,y\}} (P)
	$$
	for each light $P^* \in \mathcal{P}_{i-1,x}^H$.
	
	This is done as follows. First, each $z \in \{x,y\}$ sends to the leader of each light part $P \in \mathcal{P}_{i-1,z}^T$ its sketch $\Sketch_{G \setminus \{x,y\}} (P)$ and its outgoing edge $(u_P, v_P)$. This information is then broadcasted over each light tail-part, in parallel, using \Cref{cl:light-communication}.
	Then, the $u_P$ vertices of the light tail-parts send this information to their $v_P$ neighbors. At this point, for each light star center $P^* \in \mathcal{P}_{i-1,x}^H$, the sketch of each light tail-part $P \in \mathcal{P}_{i-1}$ which points at $P^*$ is held by exactly one vertex $v_P \in P^*$ (the same vertex can be holding several such sketches). Thus, each $v \in P^*$ can locally compute
	$$
	val(v) = \bigoplus_{\substack{P \in \mathcal{P}_{i-1,x}^T \cup \mathcal{P}_{i-1,y}^T \\ \text{$P$ is light and $v = v_P$}}} \Sketch_{G \setminus \{x,y\}} (P).
	$$
	These values are summed on each light star center $P^* \in \mathcal{P}_{i-1,x}^H$ in parallel using \Cref{cl:light-communication}, letting all vertices of $P^*$ learn $\bigoplus_{v \in P^*} val(v) = val(P^*)$. The leaders of the light star centers $P^* \in \mathcal{P}_{i-1,x}^H$ can now send this information to $x$, as required.
\end{proof}

\begin{figure}
	\begin{center}
		\includegraphics[height=7cm]{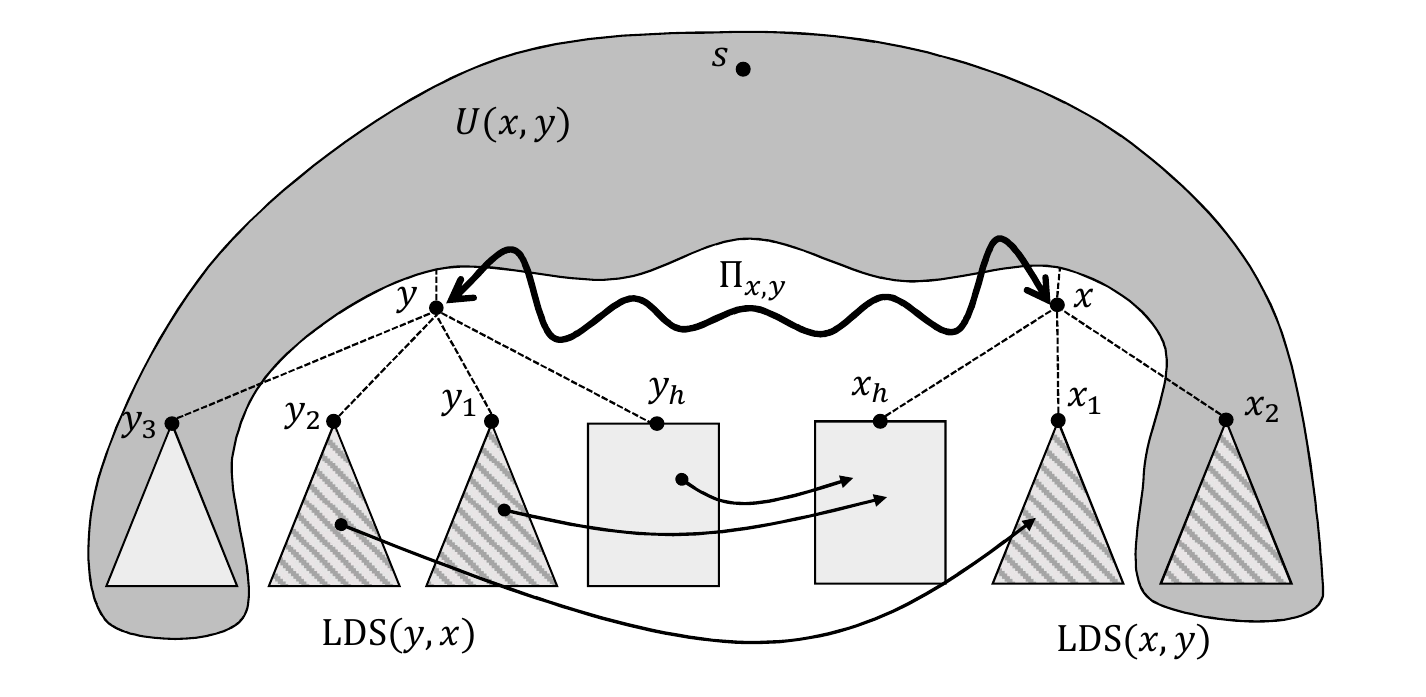}
		\caption{ \sf Illustration of the first Bor\r{u}vka phase in algorithm $\mathcal{A}^P_{x,y}$.
			Here, we assume each subtree of $x$ and $y$ forms a distinct component in $\mathcal{C}_x$ and $\mathcal{C}_y$, respectively.
			The components in $\LDS(x,y) \cup \LDS(y,x)$ are striped triangles.
			The heavy components are squares.
			The non-sensitive or pseudo-sensitive components are part of the $s$-part $U(x,y)$.
			The solid arrows correspond to the outgoing edges obtained from tail parts, pointing to head parts.
			The sketches of components $C_{y,y_2}$ and $C_{x,x_1}$ are merged via communication of the outgoing edge connecting them.
			The sum-of-sketches of ${C_{y,y_1}}$ and $H_{y}$ is sent from $y$ to $x$ on the $xy$-channel $\Pi_{x,y}$, allowing it to merge them with $H_x$.
			\label{fig:promise} 
		}
	\end{center}
\end{figure}

Finally, we show how conditions (I1-4) are satisfied w.r.t.\ $\mathcal{P}_i$. 
Condition (I2) is immediate by \Cref{cl:merges}. 
We next show conditions (I1) and (I3). In order for $x,y$ to learn the part-IDs of the non-light parts in the new partition $\mathcal{P}_i$, it is enough for them to learn the new part-ID given to each of the (at most three) non-light parts in the previous partition $P \in \mathcal{P}_{i-1}$. The vertex $z \in \{x,y\}$ which is responsible for $P$ knows its new-part ID by \Cref{cl:detect-legal-IDs} (and knows that $P$ is the $s$-part, $x_h$-part and/or $y_h$-part by (I1) and (I3) w.r.t.\ $\mathcal{P}_{i-1}$), so $z$ can share this over the $xy$-channel. Now, the $(G \setminus \{x,y\})$-sketches of each non-light part in the new partition $\mathcal{P}_i$ is known to either $x$ or $y$ by \Cref{cl:merges}, and sharing them over the $xy$-channel yields (I1) and (I3) w.r.t.\ $\mathcal{P}_i$.
Lastly, we show condition (I4). Let $z \in \{x,y\}$, let $z'$ be a child of $z$. By (I4) w.r.t.\ $\mathcal{P}_{i-1}$, $z$ knows to which part $z'$ belongs in $\mathcal{P}_{i-1}$.
If this part is non-light, then the new part-ID (of the part containing $z'$ in $\mathcal{P}_i$) is known by the previous step handling (I1) and (I3). It remains to handle the case where $z'$ belongs to a light part in $\mathcal{P}_{i-1}$. For this, we let $x,y$ inform the leaders of light parts under their responsibility in $\mathcal{P}_{i-1}$ of their new part-ID in $\mathcal{P}_i$ (which $x,y$ know by \Cref{cl:detect-legal-IDs}), and this information is broadcasted in parallel on all light parts in $\mathcal{P}_{i-1}$ using \Cref{cl:light-communication}. This informs each such $z'$ of its part-ID, and by sending it to $z$ we establish (I4) w.r.t. $\mathcal{P}_{i}$.

We are now ready to conclude the description of the algorithm $\mathcal{A}_{x,y}^P$ and prove \Cref{lem:connectivity-with-promise}.
\begin{proof}[Proof of \Cref{lem:connectivity-with-promise}]
	We have shown that the initialization of $\mathcal{A}_{x,y}^P$ and each of the Bor\r{u}vka phases are implemented in $\widetilde{O}(D)$ rounds with $\widetilde{O}(1)$ congestion, where communication is restricted only to edges incident to $\LDS(x,y) \cup \LDS(y,x)$ and the $xy$-channel.
	Next, we assert that the number of growable part reduces by a constant factor in each phase, in expectation. Indeed, given a sketch information $\Sketch_{G \setminus \{x,y\}}(P)$ for a growable part $P$, one can infer an outgoing edge $(u,v)$ from $P$ with constant probability. In addition, with probability $1/4$ this edge is valid (i.e., $P$ is a tail part and $v$ is in a head part).
	We now deduce, by Markov's inequality, that w.h.p.\ there is no growable part after $K = O(\log n)$ phases, so $\mathcal{P}_K$ are maximal connected components in $G \setminus \{x,y\}$. As $x,y$ jointly hold the part-IDs of all parts in $\mathcal{P}_K$, they can use the $xy$-channel to determine if there is more than one such part. The theorem follows.
\end{proof}

\subsubsection{Implementation Details for \Cref{sec:promise}}\label{sec:promise-imp-details}
\APPENDCLASSIFYINGSENSITIVITY
\APPENDMODIFIEDINVARIANTS
\APPENDINITSKETCHCANCEL
\APPENDLIGHTCOMM
\APPENDCLLEGAL

\subsection{Omitting the Promise}\label{sec:omit-promise}

Our goal now is omitting the promise of \Cref{sec:promise}. To this end, we first classify the independent pairs into two types, light and heavy, defined next. In this section, it is convenient to have the distinction between unorderd and ordered pairs (which was not significant in previous sections). We denote unordered pairs by $xy$, and ordered pairs by $\langle x,y \rangle$. 

\begin{definition}[Light and Heavy Pairs]\label{def:classification-pairs}
	An ordered independent pair $\langle x,y \rangle$ is called an \emph{ordered light pair} if there exists a component $C \in \mathcal{C}_x$ for which one of the following holds:
	\begin{itemize}
		\item $\pi_x (s,C)$ contains a \emph{light} $T$-edge $(y,y')$ such that $x \in \pi_y (s,C_{y,y'})$.
		\item  $y = u_C$. (Recall that $(u_C, v_C)$ is the last edge of $\pi_x (s,C)$, with $v_C \in C$, see beginning of \Cref{sec:indep}.)
	\end{itemize}
	Equivalently, $\langle x,y \rangle$ is an ordered light pair if there is exists a fully-$y$-sensitive component $C \in \mathcal{FS}(x,y)$ for which $(y,y_h) \notin \pi_x (s,C)$.
	
	An (unordered) independent pair $xy$ is called a \emph{light pair} if at least one of its orderings $\langle x,y \rangle, \langle y,x \rangle$ is an ordered light pair. Otherwise, it is called a \emph{heavy pair}.
\end{definition}

\begin{figure}
	\begin{center}
		\includegraphics[height=9cm]{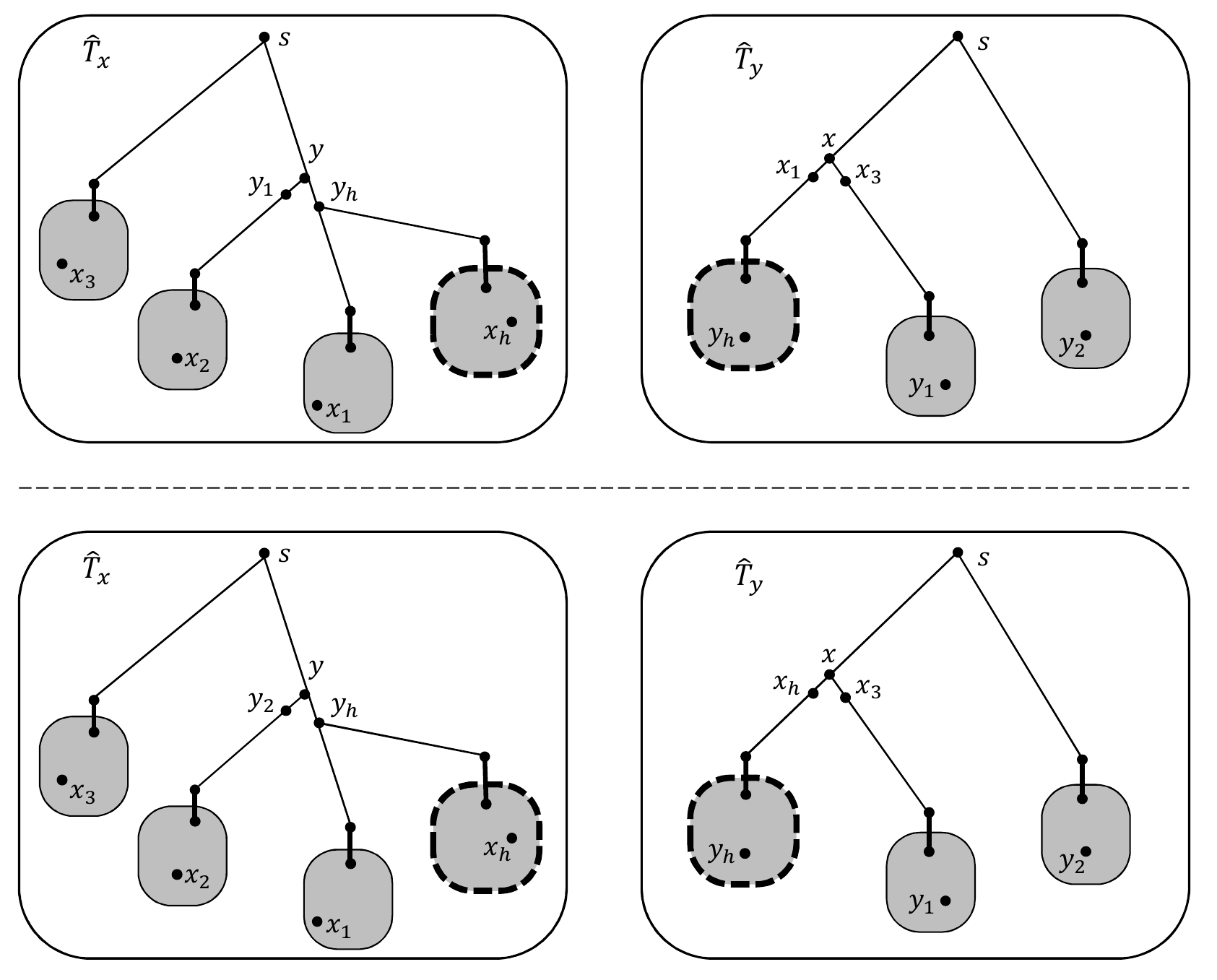}
		\caption{
        \sf Top: A light $xy$ pair.
        % $C_{x,x_2}$ is a fully-$y$-sensitive component lying below the light edge $(y,y_2)$ in $\widehat{T}_x$.
        $C_{x,x_2}$ is a fully-$y$-sensitive component lying below the light edge $(y,y_1)$ in $\widehat{T}_x$.
		Bottom: A heavy $xy$ pair. All fully-$y$-sensitive components in $\mathcal{C}_x$ lie below the heavy edge $(y,y_h)$ in $\widehat{T}_x$, and vice-versa.
		\label{fig:light-heavy-pairs}
		}
	\end{center}
\end{figure}

By negating the definition of light pairs, we immediately get:
\begin{observation}\label{obs:heavy-pairs-property}
	An independent pair $xy$ is a heavy pair iff both of the following hold:
	\begin{itemize}
		\item $C \in \mathcal{FS}(x,y) \implies (y,y_h) \in \pi_x (s,C)$, and
		\item $C \in \mathcal{FS}(y,x) \implies (x,x_h) \in \pi_y (s,C).$
	\end{itemize}
%	$$C \in \mathcal{FS}(x,y) \implies (y,y_h) \in \pi_x (s,C) \quad \text{and} \quad C \in \mathcal{FS}(y,x) \implies (x,x_h) \in \pi_y (s,C).$$
%	\begin{gather*}
%		C \in \mathcal{FS}(x,y) \implies (y,y_h) \in \pi_x (s,C) \\
%		\text{and} \\
%		C \in \mathcal{FS}(y,x) \implies (x,x_h) \in \pi_y (s,C).
%	\end{gather*}
\end{observation}

See \Cref{fig:light-heavy-pairs} for an illustration of light and heavy pairs.

\subsubsection{Light Pairs}
For every ordered light pair $\langle x,y \rangle$, we choose some arbitrary component $C^{\langle x,y \rangle} \in \mathcal{FS}(x,y)$ such that $(y,y_h) \notin \pi_x (s, C^{\langle x,y \rangle})$, which exists by \Cref{def:classification-pairs}, and define \emph{the $\langle x,y \rangle$-channel} by
%\begin{equation}\label{eq:light-xy-channel}
%	\Pi_{\langle x,y \rangle}=\pi(x,v_{C^{\langle x,y \rangle}}) \circ (v_{C^{\langle x,y \rangle}},u_{C^{\langle x,y \rangle}}) \circ \pi(u_{C^{\langle x,y \rangle}},y)~.
%\end{equation}
\begin{gather}\label{eq:light-xy-channel}
	\Pi_{\langle x,y \rangle} = 
	\pi(x,v_{C^{\langle x,y \rangle}}) \circ (v_{C^{\langle x,y \rangle}},u_{C^{\langle x,y \rangle}}) \circ \pi(u_{C^{\langle x,y \rangle}},y)~. \nonumber
\end{gather}
Let 
%$\mathcal{P}_{light} = \{\Pi_{\langle x,y \rangle} \mid \text{$\langle x, y \rangle$ is an ordered light pair} \}$.
$$
	\mathcal{P}_{light} = \{\Pi_{\langle x,y \rangle} \mid \text{$\langle x, y \rangle$ is an ordered light pair} \}.
$$
The following lemmas show that this path collection satisfies the promise, and can be efficiently learned in a distributed manner.

\begin{lemma}\label{lem:promise-light}
	The path collection $\mathcal{P}_{light}$ satisfies the promise. That is, each path in $\mathcal{P}_{light}$ has length $O(D)$, and each $G$-edge appears in at most $\widetilde{O}(D)$ many paths from $\mathcal{P}_{light}$.
\end{lemma}
\begin{proof}
	The length bound is immediate. For the second condition, it is enough to prove that every vertex in $G$ appears \emph{as an internal vertex} in at most $\widetilde{O}(D)$ many paths from $\mathcal{P}_{light}$. 
	We first observe that for an internal vertex $a \in \Pi_{\langle x,y \rangle}$, the following hold by the properties of $C^{\langle x,y \rangle}$:
	\begin{enumerate}
		\item If $a \in \pi(x,v_{C^{\langle x,y \rangle}})$, then $x$ is an ancestor of $a$, and $y \in \LA(u_{C_{x,a}}) \cup \{u_{C_{x,a}}\}$ (as $C_{x,a} = C^{\langle x, y \rangle}$).
		\item If $a \in \pi(u_{C^{\langle x,y \rangle}}, y)$, then $y \in \LA(a)$, and $x \in \pi_y(s, C_{y,a})$ (as $C_{y,a} = C_{y,y'}$, where $(y,y') \in \pi_x (s, C^{\langle x, y \rangle})$ is a light $T$-edge).
	\end{enumerate}
	By the first observation, each vertex $a$ can be internal in $\Pi_{\langle x, y \rangle}$ and appear on the first segment for $O(D)$ different $x$ vertices, and for $O(\log n)$ different $y$ vertices for each such $x$. From the second observation, $a$ can be internal in $\Pi_{\langle x, y \rangle}$ and appear in the last segment for $O(\log n)$ different $y$ vertices, and $O(D)$ different $x$ vertices for each such $y$.
	So, overall $a$ can be an internal vertex in at most $O(D \log n)$ paths from $\mathcal{P}_{light}$.
\end{proof}

\begin{lemma}\label{lem:light-channels}
	There is a randomized $\widetilde{O}(D)$-round algorithm such that w.h.p., for every ordered light pair $\langle x, y \rangle$, both $x$ and $y$ are informed of their pairing and of the vertices $v_{C^{\langle x,y \rangle}}, u_{C^{\langle x,y \rangle}}$. Consequently, $x$ and $y$ can route messages to each other on $\Pi_{\langle x,y \rangle}$.
\end{lemma}
\begin{proof}
	First, observe that each vertex $x$ can identify all vertices $y$ such $\langle x,y \rangle$ is a light ordered pair, and choose the corresponding component $C^{\langle x,y \rangle}$, by using the information of its connectivity tree $\widehat{T}_x$ and of Lemma \ref{lem:learning-strong-sens}. The vertices $v_{C^{\langle x,y \rangle}}, u_{C^{\langle x,y \rangle}}$ can also be identified by $x$ from $\pi_x (s, C^{\langle x,y \rangle})$. Hence, our goal is to supply the $y$ vertices of each ordered light pair $\langle x,y \rangle$ with all the required information, which is the list $I_{\langle x,y \rangle} = [\langle x,y \rangle, v_{C^{\langle x,y \rangle}}, u_{C^{\langle x,y \rangle}}]$ consisting of $\widetilde{O}(1)$ bits (recall that each vertex is identified with its compressed-path, i.e.\ its $\LCALabel_T$-label). We do this by routing the message $I_{\langle x,y \rangle}$ from the source $x$ to target $y$ along $\Pi_{\langle x,y \rangle}$. Note that any vertex in $\Pi_{\langle x,y \rangle}$, upon receiving the message $I_{\langle x,y \rangle}$, can determine the next edge of $\Pi_{\langle x,y \rangle}$ along which $I_{\langle x,y \rangle}$ should be routed, by using the $\LCALabel_T$-labels found in $I_{\langle x,y \rangle}$. Thus, we essentially have an instance of the packet-routing problem, where the source-target pairs, are the ordered light pairs $\langle x,y \rangle$, the packets are the messages $I_{\langle x,y \rangle}$, and the paths are the channels $\Pi_{\langle x,y \rangle}$ respectively. The dilation and congestion are both $\widetilde{O}(D)$ by \Cref{lem:promise-light}. Hence, we can use the scheduling of \Cref{thm:delay} to execute this packet-routing instance within $\widetilde{O}(D)$, w.h.p.
\end{proof}

Finally, by combining \Cref{lem:promise-light,lem:light-channels} with \Cref{cor:allpairsconnectivity-with-promise}, we get that within $\widetilde{O}(D)$-rounds, w.h.p.\ all ordered light pairs (and thus, also the unordered light pairs) can determine if they form a cut pair in $G$.

\subsubsection{Heavy Pairs}

We now consider the more challenging task of connectivity checking for heavy pairs. Our strategy is based on classifying these pairs further as \emph{mutual} and \emph{non-mutual}. The special configuration of mutual pairs is dealt with in a similar fashion to light pairs, namely, implementing the promise for the algorithm of \Cref{sec:promise} by choosing bounded-congestion communication channels. Perhaps surprisingly, for the remaining non-mutual pairs it is sufficient for the vertices to collect a small amount of information over the tree $T$, which enables every non-mutual pair $xy$ to determine the connectivity of $G \setminus \{x,y\}$ by local computation in $x$ or $y$.

We start with some preliminary notions required for our definition of mutual pairs. For a vertex $x \in V$, let $\Gates(x) = \{v \in V \setminus V(T_x) \mid \text{$v$ has a neighbor in $H_x$}\}$, and let $\lcagates(x)$ be the lowest common ancestor of the vertices in $\Gates(x)$ (w.r.t.\ $T$). Observe that for the last edge $(u_{H_x}, v_{H_x})$ of $\pi_x (s, H_x)$ we have $u_{H_x} \in V \setminus V(T_x)$ and $v_{H_x} \in H_x$, hence $u_{H_x} \in \Gates(x)$. As $\lcagates(x)$ is an ancestor of any vertex in $\Gates(x)$, we obtain that $\lcagates(x) \in \pi(s, u_{H_x})$.

\begin{definition}
	A heavy pair $xy$ is called a \emph{mutual pair} if $\lcagates(x) = y$ and $\lcagates(y) = x$. Otherwise, $xy$ is a \emph{non-mutual pair}.
\end{definition}

The following claim concerns the collection of preliminary information required for the connectivity algorithm of both mutual and non-mutual pairs.

\begin{claim}\label{cl:lcagates}
	Within $\widetilde{O}(D)$ rounds, w.h.p.\ each $x \in V$ learns $\lcagates(x)$ and the following information for every $y \in V(\widehat{T}_x)$: (i) $\lcagates(y)$, (ii) $\pi_y^*(s, H_y)$, and (iii) $\Sketch_{G \setminus \{y\}} (H_y)$.
\end{claim}
\def\APPENDLCAGATES{
% \begin{proof}[Proof of \Cref{cl:lcagates}]
\begin{proof}
	We first describe the procedure letting each $x$ learn $\lcagates(x)$.
	For a (non-empty) $U \subseteq V$, denote by $\LCA(U)$ the lowest common ancestor of all vertices in $U$. We first importantly note that $\LCA$ is an aggregate function. Also, note that given the compressed-paths of the vertices in $U$, the compressed-path of $\LCA(U)$ can be deduced. Next, observe that a vertex $a$ can use the known compressed-paths of its $G$-neighbors to locally determine which of them belong to $V \setminus V(T_x)$, for each of its ancestors $x \in \pi(s,a)$.
	Hence, by executing bottom-up $\LCA$-aggregation on $T$ and using standard pipeline techniques, within $\widetilde{O}(D)$ rounds each vertex $a$ can learn, for each of its ancestors $x \in \pi(s,a)$, the compressed-path of 
	$$
	\ell_{x,a} := \LCA(\{v \in V \setminus V(T_x) \mid \text{$v$ has a neighbor in $V(T_a)$}\}).
	$$
	Within another $\widetilde{O}(1)$ rounds, each vertex $x$ can learn, for each of its $T$-children $x'$, the compressed-path of $\ell_{x,x'}$. By \Cref{lem:connecitivity-in-Gxs}, $x$ can detect which of its children $x'$ belongs to $H_x$, and thus it can locally compute the compressed path of $\lcagates(x)$, since it holds that
	$$
	\lcagates(x) = \LCA(\{\ell_{x,x'} \mid \text{$x'$ is a $T$-child of $x$ inside $H_x$}\}).
	$$
	This compressed-path enables $x$ to locate $\lcagates(x)$ in the tree $\widehat{T}_x$ (stored locally at $x$ by \Cref{lem:path-sketch}).
	
	To finish the proof, it remains to show the procedure letting each $x$ learn the information $I_y = [\lcagates(y), \pi_y^*(s, H_y), \Sketch_{G \setminus \{y\}} (H_y)]$ for every $y \in V(\widehat{T}_x)$.
	We first observe that each vertex $y$ holds $I_y$: The vertex $\lcagates(y)$ is known to $y$ by the previous part of the proof. $\pi_y^*(s, H_y)$ is known from $\widehat{T}_y$ (by \Cref{lem:path-sketch}). $\Sketch_{G \setminus \{y\}} (H_y)$ is obtained from $\Sketch_{G} (H_y)$ by \Cref{lem:cancel-sketch-property}.
	We let every vertex $y$ exchange $I_y$ with its neighbors. By propagating this information down on $T$ in a pipeline manner, each vertex $v$ learns $I_y$ for every $y \in \pi(s,v)$, within $\widetilde{O}(D)$ rounds. Finally, we run again Step 2 of the procedure for computing connectivity trees $\widetilde{T}_x$ in \Cref{sec:connectivity-trees} (with the same random seeds), only now augmenting the path-sketches further by including next to each stored path $\pi(s,u)$ the additional $\widetilde{O}(D)$ information of $I_y$ for each $y\in \pi(s,u)$. This lets $x$ learn $I_y$ for all vertices $y$ in the paths included in $\widehat{T}_x$ (which are obtained by the path-sketches), as required.
\end{proof}
}\APPENDLCAGATES

\paragraph{Handling Mutual Pairs.}
The mutual pairs $xy$ are handled almost exactly as the light pairs, by defining $xy$ channels satisfying the promise for the connectivity algorithms of \Cref{sec:promise}. 

First, notice that for a mutual pair $xy$, both $x$ and $y$ know the identity of their mutual-mate by \Cref{cl:lcagates}. To break the symmetry between $x$ and $y$, we assume w.l.o.g.\ that $x$ has the smaller ID, and define the $xy$ channel as
$$ \Pi_{x,y} = \pi(x, v_{H_x})  \circ (v_{H_x}, u_{H_x}) \circ \pi (u_{H_x}, y). $$
Let $\mathcal{P}_{mutual} = \{\Pi_{x,y} \mid \text{$xy$ is a mutual pair} \}$. Note that the mutual pairs form a matching in the sense, each vertex $x$ appears in at most one mutual pair (i.e., $xr(x)$).

We assert that $\mathcal{P}_{mutual}$ satisfies the promise. The fact that each path has length $O(D)$ is clear. Next, consider any vertex $a$. Then $a$ can appear in the first segment of a $\mathcal{P}_{mutual}$-path $\Pi_{x,y} = \Pi_{x,\lcagates(x)} = \Pi_{\lcagates(y), y}$ only if $x \in \pi(s,a)$, and in the second segment only if $y \in \pi(s,a)$. Thus, each vertex $a$ can appear in at most $O(D)$ many $\mathcal{P}_{mutual}$-paths, hence this is also true for edges.

By applying a procedure identical to \Cref{lem:light-channels} (only replacing $u_{C^{\langle x,y \rangle}}$ and $v_{C^{\langle x,y \rangle}}$ there with $u_{H_x}$ and $v_{H_x}$, respectively), we make sure that each mutual pair $xy$ can route messages on their channel $\Pi_{x,y}$. \Cref{cor:allpairsconnectivity-with-promise} then applies with $Q$ being the set of all mutual pairs, hence within $\widetilde{O}(D)$ rounds, w.h.p.\ each such $xy$ pair determines the connectivity of $G \setminus \{x,y\}$.

\paragraph{Handling Non-Mutual Pairs.}
Finally, we turn to handle the remaining non-mutual pairs. We will show that for each such pair $xy$, one of $x,y$ can apply a local computation for determining the connectivity of $G \setminus \{x,y\}$. To this end, we start with an observation regarding heavy pairs.

\begin{observation}\label{obs:heavy-pairs-connectivity}
	Let $xy$ be a heavy pair. Then $G \setminus \{x,y\}$ is connected iff there is some $z_h \in \{x_h, y_h\}$ that is connected to the source $s$ in $G \setminus \{x,y\}$.
\end{observation}
\begin{proof}
	The `only if' direction is trivial.  We show the `if' direction for $z_h = y_h$ (the case $z_h = x_h$ is symmetric). By \Cref{lem:non-or-pseudo-sensitive}, the only vertices that can disconnect from $s$ in $G \setminus \{x,y\}$ are those of components in $\mathcal{FS}(x,y) \cup \mathcal{FS}(y,x)$.
	
	Consider first a component $C \in \mathcal{FS}(x,y)$. By \Cref{obs:heavy-pairs-property}, it holds that $(y,y_h) \in \pi_x(s, C)$. Hence, in $G \setminus \{x,y\}$, it is possible to get from $s$ to $C$ in $G \setminus \{x,y\}$ by first walking to $y_h$, and then continuing along $\pi_x(s, C)$. We can now conclude that all vertices in $V_x$ are connected to $s$ in $G \setminus \{x,y\}$.
	
	Finally, consider a component $C \in \mathcal{FS}(y,x)$. By \Cref{obs:heavy-pairs-property}, it holds that $(x,x_h) \in \pi_y(s, C)$. Then, as $x_h \in V_x$, it is possible to get from $s$ to $C$ in $G \setminus \{x,y\}$ by first walking to $x_h$, and then continuing along $\pi_y(s, C)$.
\end{proof}

Fix some vertex $x$. By the observation above, if $xy$ is a heavy-pair such that $G \setminus \{x,y\}$ is disconnected, then it must be that $y \in \pi_x (s,H_x)$. By \Cref{lem:light-channels}, $x$ knows all of its light-mates, hence it can determine all its heavy-mates $y \in \pi_x (s,H_x)$. We will show that for any heavy-mate $y \in \pi_x (s,H_x)$ such that $\lcagates(x) \neq y$, $x$ can determine the connectivity of $G \setminus \{x,y\}$ by local computation (after $\widetilde{O}(D)$ rounds of common preprocessing for all vertices). This implies that at least one vertex of every non-mutual pair can determine if it is a cut pair, as desired.

First, we assert that $x$ can disregard all heavy-mates $y$ lying (strictly) below $\lcagates(x)$ in $\pi_x (s,H_x)$, by the following observation.
\begin{observation}\label{obs:y_below_lcagates(x)}
	Let $x \in V$. If $y$ is a vertex lying below $\lcagates(x)$ in $\pi(s, u_{H_x})$, then the source $s$ is connected to $H_x$ in $G \setminus \{x,y\}$.
\end{observation}
\begin{proof}
	First note that $y$ cannot be an ancestor of all vertices in $\Gates(x)$, since this would contradict $\lcagates(x)$ being the lowest such vertex. Hence, there exists some $a \in \Gates(x)$ such that $y \notin \pi(s,a)$. By definition of $\Gates(x)$, $a \in V \setminus V(T_x)$, and $a$ has a neighbor $b$ inside $H_x$. Thus the path $\pi(s, a) \circ (a,b)$ connects $s$ to $H_x$ in $G \setminus \{x,y\}$.
\end{proof}

We turn to handle the remaining relevant heavy-mates. Fix some heavy-mate $y$ of $x$ which is strictly above $\lcagates(x)$ in $\pi(s, u_{H_x})$ (recall that we only care now about $y \neq \lcagates(x)$). If $(y,y_h) \notin \pi_x(s, H_x)$, then by \Cref{obs:heavy-pairs-property} we have that $H_x \notin \mathcal{FS}(x,y)$, so by \Cref{lem:non-or-pseudo-sensitive} we have that $x_h$ is connected to $s$ in $G \setminus \{x,y\}$. Hence, by \Cref{obs:heavy-pairs-connectivity}, $G \setminus \{x,y\}$ connected, and we are done. So, from now on assume that $(y,y_h) \in \pi_x (s, H_x)$.  Denote by $\widehat{C}(x,y)$ the union of all components in $\mathcal{C}_x$ that are below $y_h$ in $\widehat{T}_x$, that is
$$
\widehat{C}(x,y) = \bigcup_{\substack{C \in \mathcal{C}_x \\ (y,y_h) \in \pi_x (s, C)}} C ~.
$$
We have the following key claim:

\begin{claim}\label{cl:outgoing-from-hat(C)-or-Hy}
	$G \setminus \{x,y\}$ is connected iff there is an outgoing edge from $\widehat{C}(x,y) \cup H_y$ in $G \setminus \{x,y\}$.
\end{claim}
\begin{proof}
	The `only if' direction is trivial. For the `if' direction, first note that the vertices of $\widehat{C}(x,y) \cup H_y$ are all in the same connected component in $G \setminus \{x,y\}$, since they are all connected to $y_h$ in $G \setminus \{x,y\}$.
	Therefore, it is enough to show that every $a \in V \setminus (\widehat{C}(x,y) \cup H_y \cup \{x,y\})$ is connected to $s$ in $G \setminus \{x,y\}$. If $a$ is not a descendant of $x$ or $y$ in $T$ then this is obvious. If $a$ is a descendant of $x$, then it belongs to some component $C \in \mathcal{C}_x$ such that $(y,y_h) \notin \pi_x (s,C)$, so by \Cref{obs:heavy-pairs-property} (as $xy$ is a heavy pair) we have $C \notin \mathcal{FS}(x,y)$, and by \Cref{lem:non-or-pseudo-sensitive} $C$ is connected to $s$ in $G \setminus \{x,y\}$. It remains to consider $a$ which is a descendant of $y$, in some component $C \in \mathcal{C}_y$, $C \neq H_y$. It suffices to show that $(x,x_h) \notin \pi_y (s,C)$, as this would imply that $C$ and $s$ are connected in $G \setminus \{x,y\}$ by using \Cref{obs:heavy-pairs-property} and \Cref{lem:non-or-pseudo-sensitive} (similarly to the previous case). Assume towards a contradiction that $(x,x_h) \in \pi_y (s,C)$. Then for the last edge $(u_C, v_C)$ of $\pi_y (s,C)$, we have $u_C \in H_x$ and $v_C \in V \setminus V(T_x)$. Hence, $v_C \in \Gates(x)$, and thus $\lcagates(x)$ is an ancestor of $v_C$ in $T$. As $v_C \notin V(T_{y_h})$ (since $C \neq H_y$), it follows that $\lcagates(x) \notin V(T_{y_h})$. But this is a contradiction: as $(y,y_h) \in \pi(s, u_{H_x})$ and $y$ is strictly above $\lcagates(x)$ in this path, we have that $\lcagates(x) \in V(T_{y_h})$.
\end{proof}

By this last claim, in order to determine the connectivity of $G \setminus \{x,y\}$, it is enough for $x$ to have $\Sketch_{G \setminus \{x,y\}} (\widehat{C}(x,y) \cup H_y)$. Indeed, by using \Cref{lem:sketch-property} with $O(\log n)$ fresh basic sketch units\footnote{Recall that each sketch contains $L=c\log n$ basic sketch units. Hence, by taking $c$ to be a sufficiently large constant, we can guarantee that $O(\log n)$ fresh basic sketch units exist.}, $x$ can determine w.h.p.\ if such an outgoing edge exists. The only obstacle is canceling the edges of $x,y$ from the sketches, as $\Sketch_{G \setminus \{x\}} (\widehat{C}(x,y))$ can be easily computed locally in $x$, and $\Sketch_{G \setminus \{y\}} (H_y)$ is known to $x$ by \Cref{cl:lcagates}. The following technical lemma gives a procedure allowing $x$ to overcome this obstacle.

\begin{lemma}\label{lem:cancel-xy}
	Within $\widetilde{O}(D)$ rounds, each vertex $x \in V$ can learn $\Sketch_{G \setminus \{x,y\}} (\widehat{C}(x,y) \cup H_y)$ for each heavy-mate $y$ of $x$ lying above $\lcagates(x)$ in $\pi(s, u_{H_x})$.
\end{lemma}
\def\APPENDCANCELXY{
% \begin{proof}[Proof of \Cref{lem:cancel-xy}]
\begin{proof}
	Note that 
%	$$\Sketch_{G \setminus \{x,y\}} (\widehat{C}(x,y) \cup H_y) = \Sketch_{G \setminus \{x,y\}} (\widehat{C}(x,y) \setminus H_x) \oplus \Sketch_{G \setminus \{x,y\}} (H_x) \oplus \Sketch_{G \setminus \{x,y\}} (H_y) ~.$$
	\begin{gather*}
		\Sketch_{G \setminus \{x,y\}} (\widehat{C}(x,y) \cup H_y) = \\
		\Sketch_{G \setminus \{x,y\}} (\widehat{C}(x,y) \setminus H_x)
		 \oplus \Sketch_{G \setminus \{x,y\}} (H_x) 
		 \oplus \Sketch_{G \setminus \{x,y\}} (H_y) ~.
	\end{gather*}
	
	We show how the three terms on the right hand side can be computed by every $x$, for each heavy-mate $y$ above $\lcagates(x)$ in $\pi(s, u_{H_x})$.
	
	For the first term, we start by giving an $\widetilde{O}(D)$-round procedure letting every vertex $v \in V$ learn $\Sketch_{G \setminus \{y\}} (V(T_v))$ for every $y \in \bigcup_{x \in \LA(v)} \pi_x (s, C_{x,v})$. Denote by $x_i (v)$ the $i^{th}$ highest vertex in $\LA(v)$, and by $y_{ij} (v)$ the $j^{th}$ highest vertex in $\pi_{x_i (v)} (s, C_{x_i (v), v})$. Let $S(v)$ be a matrix where the $ij$ entry, for $1 \leq i \leq O(\log n)$, $1 \leq j \leq O(D)$, is the bitstring of length $\ell = \widetilde{O}(1)$ defined by
	$$
	[S(v)]_{ij} = \begin{cases}
		\text{if $y_{ij} (v)$ is defined:} & \Sketch_{G \setminus \{y_{ij} (v)\}} (v)\\
		\text{otherwise:} & 0^{\ell} \text{ (string of $\ell$ zeros) }
	\end{cases}
	$$
	By \Cref{lem:comp-path-sketch,lem:cancel-sketch-property}, each vertex $v$ can locally compute $S(v)$. Next, observe that if $y_{ij} (v)$ is defined, then for every descendant $u$ of $v$, we have $y_{ij} (u) = y_{ij} (v)$, and therefore
%	$$
%	\bigoplus_{u \in V(T_v)} [S(u)]_{ij} = \bigoplus_{u \in V(T_v)} \Sketch_{G \setminus \{y_{ij} (v)\}} (u)  = \Sketch_{G \setminus \{y_{ij} (v)\}} (V(T_v)) ~.
%	$$
	\begin{gather*}
		\bigoplus_{u \in V(T_v)} [S(u)]_{ij} = 
		\bigoplus_{u \in V(T_v)} \Sketch_{G \setminus \{y_{ij} (v)\}} (u) = 
		\Sketch_{G \setminus \{y_{ij} (v)\}} (V(T_v)) ~.
	\end{gather*}
	Thus, XOR-aggreagting the $S(u)$ matrices over subtrees yields the required information for each vertex $v$. Since each such matrix consists of $\widetilde{O}(D)$ bits, this can be done within $\widetilde{O}(D)$ rounds using standard pipeline techniques.  Within another $\widetilde{O}(D)$ rounds, each vertex can send its aggregated matrix to its parent.
	At this point, every $x \in V$ holds $\Sketch_{G \setminus \{y\}} (V(T_{x'}))$ for every light $T$-child $x'$ of $x$, and $y \in \pi_x (s, C_{x,x'})$. Using \Cref{lem:cancel-sketch-property,lem:connecitivity-in-Gxs} and the connectivity tree $\widehat{T}_x$ now allows $x$ to remove its own edges from the relevant subtree sketches and add them up to locally compute $ \Sketch_{G \setminus \{x,y\}} (\widehat{C}(x,y) \setminus H_x)$ for any heavy-mate $y$ of interest.
	
	We now handle the second term. Recall that $x$ can easily locally compute $\Sketch_{G \setminus \{x\}} (H_x)$ (as shown e.g.\ in the proof of \Cref{cl:lcagates}). Also, if $y$ lies (strictly) above $\lcagates(x)$ in $\pi(s, u_{H_x})$, then $y$ cannot have a neighbor in $H_x$ (as $\lcagates(x)$ is an ancestor of every $v \in V(T_x)$ that has neighbor in $H_x$). Thus, for every heavy-mate $y$ of interest we have $\Sketch_{G \setminus \{x\}} (H_x) = \Sketch_{G \setminus \{x,y\}} (H_x)$.
	
	Finally, we deal with the third term. By item 3 of \Cref{lem:connecitivity-in-Gxs}, $x$ knows for each of its neighbors $u$ and for each $y \in \pi(s,u)$ the component-ID of $C_{y,u}$. Thus, for every heavy-mate $y$ of interest, $x$ can identify all its edges going into $H_y$, and cancel them from $\Sketch_{G \setminus \{y\}} (H_y)$ (which is known to $x$ by \Cref{cl:lcagates}) to obtain  $\Sketch_{G \setminus \{x,y\}} (H_y)$, as required.
\end{proof}
}\APPENDCANCELXY

This concludes the $\widetilde{O}(D)$-round algorithm for detecting w.h.p. all independent cut pairs. By \Cref{sec:depend}, within another $\widetilde{O}(D)$ rounds, we can also detect w.h.p. all dependent cut pairs. \Cref{thm:distributed-cut-pairs} follows.

\section{Conclusion}
In this work, we provide distributed \congest\ algorithms for detecting vertex cuts of size at most two, with round complexity of $\widetilde{O}(D)$.
These bounds are nearly tight, and in addition, nearly match the current bounds obtained for the corresponding edge cut problems.
It would be interesting if our techniques, e.g., for eliminating the dependency in the maximum degree $\Delta$, could be extended to detecting vertex cuts of any constant size within $\poly(D, \log n)$ rounds. The latter bound is currently only achievable for the (simpler) edge cut problem \cite{parter2019small}.

\paragraph{Acknowledgments.}
We thank the anonymous reviewers of Distributed Computing for their insightful comments and suggestions that considerably improved the presentation of this work.

\bibliographystyle{alpha}
\bibliography{dist-cut}

\end{document}